\documentclass[fontsize=10pt]{scrartcl}
\usepackage[utf8]{inputenc}
\usepackage[T1]{fontenc}
\usepackage[american]{babel}
\usepackage[autostyle,english=american]{csquotes}
\usepackage[theoremfont]{newpxtext}
\usepackage{newpxmath}

\usepackage{amsthm}
\usepackage{amssymb}
\usepackage{amsmath}
\usepackage{enumitem}
\usepackage[headsepline]{scrlayer-scrpage}
\usepackage[pdfborder={0 0 0.5},pdfpagelabels,colorlinks,linkcolor=blue]{hyperref}
\newtheorem{theorem}{Theorem}
\newtheorem{lemma}[theorem]{Lemma}
\newtheorem{proposition}[theorem]{Proposition}
\newtheorem{corollary}[theorem]{Corollary}
\theoremstyle{definition}
\newtheorem{definition}[theorem]{Definition}
\newtheorem{remark}[theorem]{Remark}
\newtheorem{condition}[theorem]{Condition}
\numberwithin{theorem}{section}
\allowdisplaybreaks
\newcommand{\f}{f^\alpha}
\newcommand{\g}{g^\alpha}
\renewcommand{\a}{a^\alpha}
\newcommand{\e}{e_\alpha}
\renewcommand{\b}{b^\alpha}
\newcommand{\K}{\mathcal{K}_\alpha}
\renewcommand{\v}{\widehat{v}_\alpha}
\newcommand{\vo}{v_\alpha^0}
\newcommand{\gamm}{\gamma_\alpha}

\newcommand{\R}{\mathbb{R}}
\newcommand{\N}{\mathbb{N}}
\newcommand{\T}{T_\bullet}
\newcommand{\I}{I_{\T}}
\newcommand{\suma}{\sum_{\alpha=1}^N}
\newcommand{\sumb}{\sum_{\alpha=1}^{N'}}
\let\originalleft\left
\let\originalright\right
\renewcommand{\left}{\mathopen{}\mathclose\bgroup\originalleft}
\renewcommand{\right}{\aftergroup\egroup\originalright}

\DeclareMathOperator{\supp}{supp}
\let\div\relax
\DeclareMathOperator{\div}{div}
\DeclareMathOperator{\curl}{curl}
\newcommand{\inte}{\mathrm{int}}
\DeclareMathOperator{\dist}{dist}
\newcommand{\kin}{{\alpha\mathrm{kin}}}
\newcommand{\lt}{\mathrm{lt}}
\newcommand{\loc}{\mathrm{loc}}
\automark{section}
\automark*{subsection}
\ihead{\leftmark}
\ohead{\thepage}
\begin{document}

	\title{Weak Solutions of the Relativistic Vlasov-Maxwell System with External Currents}
	\author{Jörg Weber\\ \textit{University of Bayreuth, 95440 Bayreuth, Bavaria, Germany}\\ \texttt{Joerg.Weber@uni-bayreuth.de}}
	\date{}
	\maketitle
	\begin{abstract}
		The time evolution of a collisionless plasma is modeled by the relativistic Vlasov-Maxwell system which couples the Vlasov equation (the transport equation) with the Maxwell equations of electrodynamics. We consider the case that the plasma consists of $N$ particle species, the particles are located in a bounded container $\Omega\subset\R^3$, and are subject to boundary conditions on $\partial\Omega$. Furthermore, there are external currents, typically in the exterior of the container, that may serve as a control of the plasma if adjusted suitably. We do not impose perfect conductor boundary conditions for the electromagnetic fields, but consider the fields as functions on whole space $\R^3$ and model objects, that are placed in space, via given matrix-valued functions $\varepsilon$ (the permittivity) and $\mu$ (the permeability). A weak solution concept is introduced and existence of global in time solutions is proved, as well as the redundancy of the divergence part of the Maxwell equations in this weak solution concept.
		
		\vspace*{3mm}
		
		\noindent\textbf{Keywords}$\;$ relativistic Vlasov-Maxwell system, nonlinear partial differential equations
		
		\vspace*{3mm}
		
		\noindent\textbf{MSC Classification:}$\;$ 35Q61, 35Q83, 82D10
	\end{abstract}
	
	\numberwithin{equation}{section}
	\section{Introduction}
	The time evolution of a collisionless plasma is modeled by the relativistic Vlasov-Maxwell system. Collisions among the plasma particles can be neglected if the plasma is sufficiently rarefied or hot. The particles only interact through electromagnetic fields created collectively. We consider the following setting: There are $N$ species of particles, all of which are located in a container $\Omega\subset\R^3$, which is a bounded domain, for example a fusion reactor. Thus, boundary conditions on $\partial\Omega$ have to be imposed. In the exterior of $\Omega$, there are external currents, for example in electric coils, that may serve as a control of the plasma if adjusted suitably. In order to model materials that are placed somewhere in space, for example (almost perfect) superconductors, we consider the permittivity $\varepsilon$ and permeability $\mu$, which are functions of the space coordinate, take values in the set of symmetric, positive definite matrices of dimension three, and do not depend on time, as given. With this assumption we can model linear, possibly anisotropic materials that stay fixed in time. We should mention that in reality $\varepsilon$ and $\mu$ will on the one hand additionally depend on the particle density inside $\Omega$ and on the other hand additionally locally on the electromagnetic fields, typically via their frequencies (maybe even nonlocally because of hysteresis). However, this would cause further nonlinearities which we avoid in this work.
	
	The unknowns are on the one hand the particle densities $\f=\f\left(t,x,v\right)$, $\alpha=1,\dots,N$, which are functions of time $t\in\R$, the space coordinate $x\in\Omega$, and the momentum coordinate $v\in\R^3$. Roughly speaking, $\f\left(t,x,v\right)$ indicates how many particles of the $\alpha$-th species are at time $t$ at position $x$ with momentum $v$. On the other hand there are the electromagnetic fields $E=E\left(t,x\right)$, $H=H\left(t,x\right)$, which depend on time $t$ and space coordinate $x\in\R^3$. The $D$- and $B$-fields are computed from $E$ and $H$ by $D=\varepsilon E$ and $B=\mu H$. We will only view $E$ and $H$ as unknowns in the following. The main assumption about $\varepsilon$ (and likewise $\mu$) in Section \ref{sec:existence} will be $\sigma\leq\varepsilon\leq\sigma'$ for some $\sigma,\sigma'>0$ in the sense of positive definiteness. This property implies that
	\begin{align*}
	E\mapsto\left(\int_{\R^3}\varepsilon E\cdot E\,dx\right)^{\frac{1}{2}}
	\end{align*}
	is a norm on $L^2\left(\R^3;\R^3\right)$, which is equivalent to the standard $L^2$-norm.
	
	The Vlasov-Maxwell system on a time interval with given final time $0<\T\leq\infty$, equipped with boundary conditions on $\partial\Omega$ and initial conditions for $t=0$, is then given by the following set of equations; we explain the appearing notation afterwards:
	\begin{subequations}\renewcommand{\theequation}{VM.\arabic{equation}}\phantomsection\makeatletter\def\@currentlabel{VM}\label{eq:WholeSystem}\makeatother
		\begin{align}
		\partial_t\f+\v\cdot\partial_x\f+\e\left(E+\v\times\mu H\right)\cdot\partial_v\f&=0&\ \mathrm{on}\ \I\times\Omega\times\R^3,\label{eq:WholeVl}\\
		\f_-&=\K\f_++\g&\ \mathrm{on}\ \gamma_{\T}^-,\label{eq:WholeBoun}\\
		\f\left(0\right)&=\mathring\f&\ \mathrm{on}\ \Omega\times\R^3,\label{eq:WholeInitf}\\
		\varepsilon\partial_tE-\curl_xH&=-4\pi j&\ \mathrm{on}\ \I\times\R^3,\label{eq:WholeMax1}\\
		\mu\partial_tH+\curl_xE&=0&\ \mathrm{on}\ \I\times\R^3,\\
		\left(E,H\right)\left(0\right)&=\left(\mathring E,\mathring H\right)&\ \mathrm{on}\ \R^3,\label{eq:WholeInitEH}
		\end{align}
	\end{subequations}
	where \eqref{eq:WholeVl} to \eqref{eq:WholeInitf} have to hold for all $\alpha=1,\dots,N$ and $\I$ denotes the given time interval $\I:=\left[0,\T\right]$ if $0\leq\T<\infty$ and $\I:=\left[0,\infty\right[$ if $\T=\infty$, respectively. Additionally, the divergence equations
	\begin{subequations}\label{eq:diveqn}
		\begin{align}
		\div_x\left(\varepsilon E\right)&=4\pi\rho&\ \mathrm{on}\ \I\times\R^3,\label{eq:diveqnD}\\
		\div_x\left(\mu H\right)&=0&\ \mathrm{on}\ \I\times\R^3,
		\end{align}
	\end{subequations}
	have to hold. In \eqref{eq:WholeInitf} and \eqref{eq:WholeInitEH}, $\f\left(0\right)$ and $\left(E,H\right)\left(0\right)$ denote the evaluation of $\f$ and $\left(E,H\right)$ at time $t=0$, that is to say the function $\f\left(0,\cdot,\cdot\right)$. We will use this notation often, also similarly for other functions.
	
	In \eqref{eq:WholeVl}, $\e$ is the charge of the $\alpha$-th particle species and $\v$ the velocity, which is computed from the momentum $v$ by
	\begin{align*}
	\v=\frac{v}{\sqrt{m_\alpha^2+\left|v\right|^2}}.
	\end{align*}
	To ensure that the speed of light is constant in $\Omega$ and hence ensure that $\v$ is independent of $x$, we have to assume that $\varepsilon\mu$ is constant in $\Omega$. Throughout this work we use modified Gaussian units such that $\varepsilon=\mu=1$ on $\Omega$ -- thus, the speed of light is $1$ in $\Omega$ -- and all rest masses $m_\alpha$ of a particle of the respective species are at least $1$. Clearly, $\left|\v\right|<1$, that is, the velocity of a particle is bounded by the speed of light (in $\Omega$).
	
	Equation \eqref{eq:WholeBoun} describes the boundary condition on $\partial\Omega$. Typically, one imposes specular boundary conditions. Thus it is natural to consider the following decompositions:
	\begin{align*}
	\tilde\gamma^\pm&:=\left\{\left(x,v\right)\in\partial\Omega\times\R^3\mid v\cdot n\left(x\right)\gtrless 0\right\},\quad\tilde\gamma^0:=\left\{\left(x,v\right)\in\partial\Omega\times\R^3\mid v\cdot n\left(x\right)=0\right\},\\
	\gamma^\pm&:=\left[0,\infty\right[\times\tilde\gamma^\pm,\quad\gamma^0:=\left[0,\infty\right[\times\tilde\gamma^0,\quad\gamma_T^\pm:=I_T\times\tilde\gamma^\pm,\quad\gamma_T^0:=I_T\times\tilde\gamma^0,
	\end{align*}
	where $n\left(x\right)$ is the outer unit normal of $\partial\Omega$ at $x\in\partial\Omega$ and $0<T\leq\infty$. In \eqref{eq:WholeBoun}, $\f_\pm$ are the restrictions of $\f$ to $\gamma_{\T}^\pm$. The operator $\K$ maps functions on $\gamma_{\T}^+$ to functions on $\gamma_{\T}^-$. In Section \ref{sec:existence} we deal with the case that
	\begin{align}\label{eq:specialK}
	\K h=\a(Kh)
	\end{align}
	where
	\begin{align*}
	\left(Kh\right)\left(t,x,v\right)=h\left(t,x,v-2\left(v\cdot n\left(x\right)\right)\right)
	\end{align*}
	describes reflection on the boundary and $\a$, satisfying $0\leq\a\leq 1$, describes how many of the particles hitting the boundary at time $t$ at $x$ with momentum $v$ are reflected (and not absorbed); $\g$ is the source term according to how many particles are added from outside. We will deal with partially absorbing ($\a\leq\a_0$ for some $\a_0<1$) and purely reflecting ($\a=1$, $\g=0$) boundary conditions.
	
	In \eqref{eq:WholeMax1} and \eqref{eq:diveqnD}, $j$ and $\rho$ are the current and charge density. Typically they are the sum of the internal current and charge densities,
	\begin{align*}
	j^\inte:=\suma\e\int_{\R^3}\v\f\,dv,\quad\rho^\inte:=\suma\e\int_{\R^3}\f\,dv
	\end{align*}
	and some external current density $u$ and charge density $\rho^u$ resulting from $u$. Usually, the divergence equations \eqref{eq:diveqn} are known to be redundant if all functions are smooth enough, local conservation of charge is satisfied, i.e.
	\begin{align*}
	\partial_t\rho+\div_xj=0,
	\end{align*}
	and \eqref{eq:diveqn} holds initially, which we then view as a constraint on the initial data. Therefore, in the first sections we ignore \eqref{eq:diveqn} and discuss in Section \ref{sec:diveqn} in what sense \eqref{eq:diveqn} is satisfied in the context of a weak solution concept. Since \eqref{eq:diveqn} has to hold on whole space $\R^3$, the main difficulty will be that we have to \enquote{cross over} $\partial\Omega$.
	
	The paper is organized as follows: In Section \ref{sec:statement} we state our main two theorems. The first regards the existence of weak solutions to \eqref{eq:WholeSystem}. In Section \ref{sec:existence} we prove this theorem. To this end, we state some basic results about linear Vlasov and Maxwell equations (Section \ref{sec:linearVM}), approximate the given functions in a proper way (Section \ref{sec:approx}), consider a cut-off system (Section \ref{sec:cutoff}), and finally remove the cut-off (Section \ref{sec:remove}). The second main result regards the redundancy of the divergence equations in our weak solution concept. We prove this theorem in Section \ref{sec:diveqn} and give some comments on the physical interpretation of the obtained equations.
	
	In the first part, we proceed similarly to Guo \cite{Guo93}, who proved existence of weak solutions in the case that $\varepsilon=\mu=1$, $u=0$, and the electromagnetic fields are subject to perfect conductor boundary conditions on $\partial\Omega$, i.e., $E\times n=0$. However, there is no need of artificially inserting the factor $e^{-t}$ as is done throughout that paper. The more important motivation of our paper is the following: The papers concerning plasma in a domain we are aware of deal with perfect conductor boundary conditions for the electromagnetic fields. Such a set-up can model no interaction between this domain and the exterior. However, considering fusion reactors, there are external currents in the exterior, for example in field coils. These external currents induce electromagnetic fields and thus influence the behavior of the internal plasma. Even more important, the main aim of fusion plasma research is to adjust these external currents \enquote{suitably}. Thus, we impose Maxwell's equations globally in space and model objects like the reactor wall, electric coils, and almost perfect superconductors via $\varepsilon$ and $\mu$.
	
	Vlasov-Maxwell systems have been studied extensively. In case of no reactor wall, i.e., the Vlasov equation is imposed globally in space (as well as Maxwell's equations), global well-posedness of the Cauchy problem is a famous open problem. Global existence and uniqueness of classical solutions has been proved in lower dimensional settings, see Glassey and Schaeffer \cite{GS90,GS97,GS98a,GS98b}. In the full three-dimensional setting, global existence of weak solutions was proved by Di Perna and Lions \cite{DL89}. Their momentum-averaging lemma is fundamental for proving existence of weak solutions in any setting (with or without boundary, with or without perfect conductor boundary conditions and so on), since it handles the nonlinearity in the Vlasov equation. However, uniqueness of these weak solutions is not known. For a more detailed overview we refer to Rein \cite{Rei04}. 
	
	\section{Preliminaries}
	\subsection{Some notation}
	In the following, we denote by $\chi_M$ the characteristic function of some set $M$ (i.e., $\chi_M\left(x\right)=1$ if $x\in M$ and $0$ otherwise) and by $\chi_T$ the characteristic function of $\left[0,T\right]$. For $1\leq p<\infty$ and $\alpha=1,\dots,N$ we define
	\begin{align*}
	L_\kin^p\left(A,da\right):=\left\{u\in L^p\left(A,da\right)\mid\int_A\vo\left|u\right|^p\,da<\infty\right\},
	\end{align*}
	equipped with the corresponding weighted norm. Here, $A\subset\R^3\times\R^3$ or $A\subset\R\times\R^3\times\R^3$ is some Borel set equipped with a measure $a$ and the weight $\vo$ is given by
	\begin{align*}
	\vo:=\sqrt{m_\alpha^2+\left|v\right|^2}.
	\end{align*}
	By $m_\alpha\geq 1$ we have $\vo\geq 1$. Moreover we write
	\begin{align*}
	L_\lt^p\left(A,da\right):=\left\{u\colon A\to\R\mid\chi_Tu\in L^p\left(A,da\right)\ \mathrm{for\ all}\ T>0\right\}
	\end{align*}
	for $1\leq p\leq\infty$. If $a$ is the Lebesgue measure we write $L_\kin^p\left(A\right)$ and $L_\lt^p\left(A\right)$, respectively. A combination $L_{\kin,\lt}^p\left(A,da\right)$ is defined accordingly. Furthermore we abbreviate
	\begin{align*}
	G_\lt\left(I;X\right):=\left\{u\colon I\to X\mid u\in G\left(0,T;X\right)\mathrm{\ for\ all\ }T\in I\right\}
	\end{align*}
	where $0\in I\subset\left[0,\infty\right[$ is some interval, $G$ is some $C^k$ or $L^p$, and $X$ is a normed vector space.
	
	For ease of notation it will be convenient to introduce a surface measure on $\left[0,\infty\right[\times\partial\Omega\times\R^3$, namely
	\begin{align*}
	d\gamm=\left|\v\cdot n\left(x\right)\right|\,dvdS_xdt.
	\end{align*}
	
	Since $\varepsilon$ is already used for the permittivity, the letter $\iota$, and not $\varepsilon$, will always denote a small positive number.
	
	For a matrix $A\in\R^{3\times 3}$ and a positive number $\sigma>0$, we write $A\geq\sigma$ ($A\leq\sigma$) if $Ax\cdot x\geq\sigma\left|x\right|^2$ ($Ax\cdot x\leq\sigma\left|x\right|^2$) for all $x\in\R^3$. For measurable $A\colon\R^3\to\R^{3\times 3}$ and $\sigma>0$ we write $A\geq\sigma$ ($A\leq\sigma$) if $A\left(x\right)\geq\sigma$ ($A\left(x\right)\leq\sigma$) for almost all $x\in\R^3$.
	
	Finally, for a normed space $X$, some $x\in X$ and $r>0$, $B_r\left(x\right)$ denotes the open ball in $X$ with center $x$ and radius $r$. Furthermore we abbreviate $B_r:=B_r\left(0\right)$.
	
	\subsection{Weak formulation}
	The space of test functions for \eqref{eq:WholeVl} to \eqref{eq:WholeInitf} will be
	\begin{align*}
	\Psi_{\T}&:=\left\{\psi\in C^\infty\left(\I\times\overline\Omega\times\R^3\right)\mid\right.&\supp\psi\subset\left[0,\T\right[\times\overline\Omega\times\R^3\ \mathrm{compact},\phantom{\,\}}\nonumber\\
	&&\dist\left(\supp\psi,\gamma_{\T}^0\right)>0,\phantom{\,\}}\nonumber\\
	&&\left.\vphantom{\left(\I\times\overline\Omega\times\R^3\right)}\dist\left(\supp\psi,\left\{0\right\}\times\partial\Omega\times\R^3\right)>0\right\}.
	\end{align*}
	On the other hand,
	\begin{align*}
	\Theta_{\T}:=\left\{\vartheta\in C^\infty\left(\I\times \R^3;\R^3\right)\mid\supp\vartheta\subset\left[0,\T\right[\times\R^3\ \mathrm{compact}\right\}
	\end{align*}
	will be the space of test functions for \eqref{eq:WholeMax1} to \eqref{eq:WholeInitEH}.
	
	We start with the definition of what we call solutions to \eqref{eq:WholeSystem}.
	\begin{definition}\label{def:WeakSolWholeSys}
		Let $0<\T\leq\infty$. We call a tuple $\left(\left(\f,\f_+\right)_\alpha,E,H,j\right)$ a weak solution of \eqref{eq:WholeSystem} on the time interval $\I$ if (for all $\alpha$)
		\begin{enumerate}[label=(\roman*)]
			\item $\f\in L_\loc^1\left(\I\times\overline\Omega\times\R^3\right)$, $\f_+\in L_\loc^1\left(\gamma_{\T}^+,d\gamm\right)$, $E,H,j\in L_\loc^1\left(\I\times\R^3;\R^3\right)$;
			\item for all $\psi\in\Psi_{\T}$ it holds that
			\begin{align}\label{eq:Vlasovweak}
			0&=-\int_0^{\T}\int_\Omega\int_{\R^3}\left(\partial_t\psi+\v\cdot\partial_x\psi+\e\left(E+\v\times H\right)\cdot\partial_v\psi\right)\f\,dvdxdt\nonumber\\
			&\phantom{=\;}+\int_{\gamma_{\T}^+}\f_+\psi\,d\gamm-\int_{\gamma_{\T}^-}\left(\K\f_++\g\right)\psi\,d\gamm-\int_\Omega\int_{\R^3}\mathring\f\psi\left(0\right)\,dvdx
			\end{align}
			(in particular, especially the integral of $\left(E+\v\times H\right)\f\cdot\partial_v\psi$ is supposed to exist);
			\item for all $\vartheta\in\Theta_{\T}$ it holds that
			\begin{subequations}\label{eq:Maxwellweak}
				\begin{align}
				0&=\int_0^{\T}\int_{\R^3}\left(\varepsilon E\cdot\partial_t\vartheta-H\cdot\curl_x\vartheta-4\pi j\cdot\vartheta\right)\,dxdt+\int_{\R^3}\varepsilon\mathring{E}\cdot\vartheta\left(0\right)\,dx,\label{eq:Maxwellweak1}\\
				0&=\int_0^{\T}\int_{\R^3}\left(\mu H\cdot\partial_t\vartheta+E\cdot\curl_x\vartheta\right)\,dxdt+\int_{\R^3}\mu\mathring{H}\cdot\vartheta\left(0\right)\,dx.\label{eq:Maxwellweak2}
				\end{align}
			\end{subequations}
		\end{enumerate}
	\end{definition}
	We easily derive this weak formulation after multiplying the respective equations of \eqref{eq:WholeSystem} with the respective test function and integrating by parts, assuming all functions are smooth enough.
	
	\subsection{Statement of main results}\label{sec:statement}
	We have two main results: The first is about existence of weak solutions in the case of partially absorbing boundary conditions for particle species $1,\dots,N'$ and purely reflecting boundary conditions for particle species $N'+1,\dots,N$. We assume that the following conditions hold:
	\begin{condition}\label{cond:data}\
		\begin{itemize}
			\item $0\leq\mathring\f\in\left(L_\kin^1\cap L^\infty\right)\left(\Omega\times\R^3\right)$ for all $\alpha=1,\dots,N$;
			\item $\K$ is given by \eqref{eq:specialK} for $\alpha=1,\dots,N$;
			\item $0\leq\a\in L^\infty\left(\gamma_{\T}^-\right)$, $\a_0:=\|\a\|_{L^\infty\left(\gamma_{\T}^-\right)}<1$, $0\leq\g\in\left(L_{\kin,\lt}^1\cap L_\lt^\infty\right)\left(\gamma_{\T}^-\right)$ for $\alpha=1,\dots,N'$;
			\item $0\leq\a\in L^\infty\left(\gamma_{\T}^-\right)$, $\|\a\|_{L^\infty\left(\gamma_{\T}^-\right)}=1$, $\g=0$ for $\alpha=N'+1,\dots,N$;
			\item $\mathring E,\mathring H\in L^2\left(\R^3;\R^3\right)$;
			\item $\varepsilon,\mu\in L^\infty\left(\R^3;\R^{3\times 3}\right)$ such that there are $\sigma,\sigma'>0$ satisfying $\sigma\leq\varepsilon,\mu\leq\sigma'$, and $\varepsilon=\mu=1$ on $\Omega$;
			\item $u\in L_{\lt}^1\left(\I;L^2\left(\Gamma;\R^3\right)\right)$.
		\end{itemize}
	\end{condition}
	Then our first main result is (see Section \ref{sec:existence}):
	\begin{theorem}\label{thm:Existence}
		Let $\T\in\left]0,\infty\right]$, $\Omega\subset\R^3$ be bounded domain such that $\partial\Omega$ is of class $C^{1,\kappa}$ for some $0<\kappa\leq 1$, and let Condition \ref{cond:data} hold. Then there exist functions
		\begin{itemize}
			\item $\f\in L_\lt^\infty\left(\I;\left(L_\kin^1\cap L^\infty\right)\left(\Omega\times\R^3\right)\right),\f_+\in\left(L_{\kin,\lt}^1\cap L_\lt^\infty\right)\left(\gamma_{\T}^+,d\gamm\right)$, $\alpha=1,\dots,N'$, all nonnegative,
			\item $\f\in L^\infty\left(\I\times\Omega\times\R^3\right)\cap L_\lt^\infty\left(\I;L_\kin^1\left(\Omega\times\R^3\right)\right),\f_+\in L^\infty\left(\gamma_{\T}^+,d\gamm\right)$, $\alpha=N'+1,\dots,N$, all nonnegative,
			\item $\left(E,H\right)\in L_\lt^\infty\left(\I;L^2\left(\R^3;\R^6\right)\right)$
		\end{itemize}
		such that $\left(\left(\f,\f_+\right)_\alpha,E,H,j\right)$ is a weak solution of \eqref{eq:WholeSystem} on the time interval $\I$ in the sense of Definition \ref{def:WeakSolWholeSys}, where
		\begin{align*}
		j&=j^\inte+u=\suma\e\int_{\R^3}\v\f\,dv+u,\quad j^\inte\in L_\lt^\infty\left(\I;\left(L^1\cap L^{\frac{4}{3}}\right)\left(\Omega;\R^3\right)\right).
		\end{align*}
		Furthermore, we have the following estimates for $1\leq p\leq\infty$ and $0<T\in\I$:\\
		Estimates on $\f,\f_+$:
		\begin{align}
		\left\|\f\right\|_{L^\infty\left(0,T;L^p\left(\Omega\times\R^3\right)\right)}&\leq\left\|\mathring \f\right\|_{L^p\left(\Omega\times\R^3\right)}+\left(1-\a_0\right)^{\frac{1}{p}-1}\left\|\g\right\|_{L^p\left(\gamma_T^-,d\gamm\right)},\label{eq:estf1}\\
		\left\|\f_+\right\|_{L^p\left(\gamma_T^+,d\gamm\right)}&\leq\left(1-\a_0\right)^{-\frac{1}{p}}\left\|\mathring \f\right\|_{L^p\left(\Omega\times\R^3\right)}+\left(1-\a_0\right)^{-1}\left\|\g\right\|_{L^p\left(\gamma_T^-,d\gamm\right)}
		\end{align}
		for $\alpha=1,\dots,N'$ and
		\begin{align}
		\left\|\f\right\|_{L^\infty\left(0,T;L^p\left(\Omega\times\R^3\right)\right)}&\leq\left\|\mathring \f\right\|_{L^p\left(\Omega\times\R^3\right)},\\
		\left\|\f_+\right\|_{L^\infty\left(\gamma_T^+,d\gamm\right)}&\leq\left\|\mathring \f\right\|_{L^\infty\left(\Omega\times\R^3\right)}
		\end{align}
		for $\alpha=N'+1,\dots,N$.\\
		Energy-like estimate:
		\begin{align}
		&\left(\sumb\left(1-\a_0\right)\int_{\gamma_T^+}\vo \f_+\,d\gamm\vphantom{\left\|\suma\int_\Omega\int_{\R^3}\vo\f\left(\cdot\right)\,dvdx+\frac{\sigma}{8\pi}\left\|\left(E,H\right)\left(\cdot\right)\right\|_{L^2\left(\R^3;\R^6\right)}^2\right\|_{L^\infty\left(\left[0,T\right]\right)}}+\left\|\suma\int_\Omega\int_{\R^3}\vo\f\left(\cdot\right)\,dvdx+\frac{\sigma}{8\pi}\left\|\left(E,H\right)\left(\cdot\right)\right\|_{L^2\left(\R^3;\R^6\right)}^2\right\|_{L^\infty\left(\left[0,T\right]\right)}\right)^{\frac{1}{2}}\nonumber\\
		&\leq\left(\suma\int_\Omega\int_{\R^3}\vo\mathring\f\,dvdx+\sumb\int_{\gamma_T^-}\vo\g\,d\gamm+\frac{\sigma'}{8\pi}\left\|\left(\mathring E,\mathring H\right)\right\|_{L^2\left(\R^3;\R^6\right)}^2\right)^{\frac{1}{2}}\nonumber\\
		&\phantom{=\;}+\sqrt{2\pi}\sigma^{-\frac{1}{2}}\left\|u\right\|_{L^1\left(0,T;L^2\left(\Gamma;\R^3\right)\right)}.
		\end{align}
		Estimate on $j^\inte$:
		\begin{align}\label{eq:estjint}
		&\left\|j^\inte\right\|_{L^\infty\left(0,T;L^{\frac{4}{3}}\left(\Omega;\R^3\right)\right)}\nonumber\\
		&\leq\left(\suma\left|\e\right|^4\left(\frac{4\pi}{3}\left\|\mathring\f\right\|_{L^\infty\left(\Omega\times\R^3\right)}+1+\begin{cases}\frac{4\pi}{3\left(1-\a_0\right)}\left\|\g\right\|_{L^\infty\left(\gamma_T^-\right)},&\alpha\leq N'\\0,&\alpha>N'\end{cases}\right)^4\right)^{\frac{1}{4}}\nonumber\\
		&\phantom{=\;}\cdot\left(\left(\suma\int_\Omega\int_{\R^3}\vo\mathring\f\,dvdx+\sumb\int_{\gamma_T^-}\vo\g\,d\gamm+\frac{\sigma'}{8\pi}\left\|\left(\mathring E,\mathring H\right)\right\|_{L^2\left(\R^3;\R^6\right)}^2\right)^{\frac{1}{2}}\right.\nonumber\\
		&\omit\hfill$\displaystyle\left.\vphantom{\left(\suma\int_\Omega\int_{\R^3}\vo\mathring\f\,dvdx+\sumb\int_{\gamma_T^-}\vo\g\,d\gamm+\frac{\sigma'}{8\pi}\left\|\left(\mathring E,\mathring H\right)\right\|_{L^2\left(\R^3;\R^6\right)}^2\right)^{\frac{1}{2}}}+\sqrt{2\pi}\sigma^{-\frac{1}{2}}\left\|u\right\|_{L^1\left(0,T;L^2\left(\Gamma;\R^3\right)\right)}\right)^{\frac{3}{2}}.$
		\end{align}
	\end{theorem}
	The second main result answers the question whether the divergence equations \eqref{eq:diveqn} are automatically satisfied if we have a weak solution of \eqref{eq:WholeSystem}. To this end, we have to introduce an external charge density $\rho^u$ corresponding to $u$ and assume that local conservation of the external charge holds:
	\begin{condition}\label{cond:extchargedens}
		There are $\rho^u\in L_\loc^1\left(\I\times\Gamma\right)$ and $\mathring\rho^u\in L_\loc^1\left(\Gamma\right)$ such that $\partial_t\rho^u+\div_xu=0$ on $\left]0,\T\right[\times\R^3$ and $\rho^u\left(0\right)=\mathring\rho^u$ on $\Gamma$, which is to be understood in the following weak sense:
		\begin{align*}
		0=\int_0^{\T}\int_{\R^3}\left(\rho^u\partial_t\psi+u\cdot\partial_x\psi\right)\,dxdt+\int_{\R^3}\mathring\rho^u\psi\left(0\right)\,dx
		\end{align*}
		for any $\psi\in C^\infty\left(\I\times\R^3\right)$ with $\supp\psi\subset\left[0,\T\right[\times\R^3$ compact. Here, $\rho^u$ and $\mathring\rho^u$ are extended by zero outside $\Gamma$.
	\end{condition}
	Then our second main result is (see Section \ref{sec:diveqn}):
	\begin{theorem}\label{thm:reddivE}
		Let $\Omega\subset\R^3$ be a bounded domain with boundary $\partial\Omega$ of class $C^1\cap W^{2,\infty}$. Furthermore let, for all $\alpha\in\left\{1,\dots,N\right\}$, $\f\in\left(L_\lt^1\cap L_{\kin,\lt}^2\cap L_\lt^\infty\right)\left(\I\times\Omega\times\R^3\right)$, $\f_+\in L_\lt^\infty\left(\gamma_{\T}^+\right)$, $\left(E,H\right)\in L_\lt^q\left(\I;L^2\left(\R^3;\R^6\right)\right)$ for some $q>2$, $\mathcal K_\alpha\colon L_\lt^\infty\left(\gamma_{\T}^+\right)\to L_\lt^\infty\left(\gamma_{\T}^-\right)$, $\g\in L_\lt^\infty\left(\gamma_{\T}^-\right)$, $\mathring\f\in\left(L^1\cap L^\infty\right)\left(\Omega\times\R^3\right)$, $\left(\mathring E,\mathring H\right)\in L^2\left(\R^3;\R^6\right)$, $\varepsilon,\mu\in L_\loc^\infty\left(\R^3;\R^{3\times 3}\right)$ with $\varepsilon=\mu=1$ on $\Omega$, and $u\in L_\loc^1\left(\I\times\Gamma;\R^3\right)$ such that the tuple $\left(\left(\f,\f_+\right)_\alpha,E,H,j^\inte+u\right)$ is a weak solution of \eqref{eq:WholeSystem} in the sense of Definition \ref{def:WeakSolWholeSys}. Furthermore, assume that Condition \ref{cond:extchargedens} holds. Moreover, let initially
		\begin{align*}
		\div_x\left(\varepsilon\mathring E\right)&=4\pi\left(\mathring\rho^\inte+\mathring\rho^u\right):=4\pi\left(\suma\e\int_{\R^3}\mathring\f\,dv+\mathring\rho^u\right),\\
		\div_x\left(\mu\mathring H\right)&=0
		\end{align*}
		on $\R^3$ be satisfied in the sense of distributions. Then:
		\begin{enumerate}[label=(\roman*)]
			\item There holds
			\begin{align*}
			\div_x\left(\mu H\right)=0
			\end{align*}
			on $\left]0,\T\right[\times\R^3$ in the sense of distributions. (This even holds under the weakest possible assumptions, i.e., all integrals in Definition \ref{def:WeakSolWholeSys} exist.)
			\item\label{thm:reddivEi} We have
			\begin{align*}
			\div_x\left(\varepsilon E\right)=4\pi\left(\rho^\inte+\rho^u\right)
			\end{align*}
			on $\left]0,\T\right[\times\left(\R^3\setminus\partial\Omega\right)$ in the sense of distributions, i.e.,
			\begin{align*}
			0=\int_0^{\T}\int_{\R^3}\left(\varepsilon E\cdot\partial_x\varphi+4\pi\left(\rho^\inte+\rho^u\right)\varphi\right)\,dxdt
			\end{align*}
			for all $\varphi\in C_c^\infty\left(\left]0,\T\right[\times\left(\R^3\setminus\partial\Omega\right)\right)$.
			\item\label{thm:reddivEii} If, additionally to the given assumptions, $\f_+\in L_\lt^1\left(\gamma_{\T}^+,d\gamm\right)$, $\g\in L_\lt^1\left(\gamma_{\T}^-,d\gamm\right)$, and $\mathcal K_\alpha\colon\left(L_\lt^1\cap L_\lt^\infty\right)\left(\gamma_{\T}^+,d\gamm\right)\to\left(L_\lt^1\cap L_\lt^\infty\right)\left(\gamma_{\T}^-,d\gamm\right)$ for all $\alpha\in\left\{1,\dots,N\right\}$, then
			\begin{align}\label{eq:divErho}
			\div_x\left(\varepsilon E\right)=4\pi\left(\rho^\inte+\rho^u+S_{\partial\Omega}\right)
			\end{align}
			on $\left]0,\T\right[\times\R^3$ in the sense of distributions, i.e.,
			\begin{align*}
			0=\int_0^{\T}\int_{\R^3}\left(\varepsilon E\cdot\partial_x\varphi+4\pi\left(\rho^\inte+\rho^u\right)\varphi\right)\,dxdt+4\pi S_{\partial\Omega}\varphi
			\end{align*}
			for all $\varphi\in C_c^\infty\left(\left]0,\T\right[\times\R^3\right)$. Here, the distribution $S_{\partial\Omega}$, satisfying $\supp S_{\partial\Omega}\subset\left]0,\T\right[\times\partial\Omega$, is given by
			\begin{align*}
			S_{\partial\Omega}\varphi&=\int_0^{\T}\int_{\partial\Omega}\varphi\left(t,x\right)\int_0^tn\left(x\right)\cdot\left(\suma\e\int_{\left\{v\in\R^3\mid n\left(x\right)\cdot v>0\right\}}\v\f_+\left(s,x,v\right)\,dv\right.\\
			&\phantom{=\;\;}\left.+\suma\e\int_{\left\{v\in\R^3\mid n\left(x\right)\cdot v<0\right\}}\v\left(\left(\mathcal K_\alpha\f_+\right)\left(s,x,v\right)+\g\left(s,x,v\right)\right)\,dv\right)\,dsdS_xdt.
			\end{align*}
		\end{enumerate}
	\end{theorem}
	Note that $\K$ need not necessarily have the structure \eqref{eq:specialK} in Theorem \ref{thm:reddivE}.
	
	\section{Existence of weak solutions}\label{sec:existence}
	In this section, we proceed similarly to Guo \cite{Guo93} with necessary modifications being made, who considered the problem with $\varepsilon=\mu=1$, $u=0$, and perfect conductor boundary conditions for the electromagnetic fields on $\partial\Omega$. Citations of this paper always refer to the relativistic version of the respective lemma, theorem etc., see \cite[Section 5]{Guo93}.
	
	\subsection{Results about linear Vlasov and Maxwell equations}\label{sec:linearVM}
	The strategy is to consider an iteration scheme where we decouple Vlasov's equations from Maxwell's equations in each iteration step and hence only have to solve linear problems. Thus it is natural to consider linear Vlasov and Maxwell equations first. Regarding the Vlasov part, we refer to Beals and Protopopescu \cite{BP87}. Considering the linear problem (on some $\left[0,T\right]$)
	\begin{subequations}\label{eq:linearVlasov}
		\begin{align}
		Yf:=\partial_t f+\v\cdot\partial_x f+F\cdot\partial_v f&=0,\\
		f_-&=\mathcal K f_++g,\\
		f\left(0\right)&=\mathring f,
		\end{align}
	\end{subequations}
	with a Lipschitz continuous, bounded force field $F$, that is divergence free with respect to $v$, they introduced a space of test functions associated to $F$. As in \cite[Lemma 2.1.]{Guo93} we can show that our test function space $\Psi_T$ belongs to that test function space for each $F$ and $T$, where one needs the assumption that $\partial\Omega$ be of class $C^{1,\kappa}$ and that the support of any $\psi\in\Psi_T$ be away from $\gamma_T^0$ and $\left\{0\right\}\times\partial\Omega\times\R^3$. In \cite{BP87}, \enquote{strong} solutions in a set of $L^p$-functions for which a trace on the boundary exists in the sense of the following extended Green's identity were searched for:
	\begin{align*}
	\int_0^T\int_\Omega\int_{\R^3}\left(\phi Yf+fY\phi\right)dvdxdt=\int_{D_T^+}f^+\phi\,d\nu^+-\int_{D_T^-}f^-\phi\,d\nu^-,
	\end{align*}
	which is supposed to hold for all test functions $\phi$. Here, $D_T^\pm$ are the outgoing/incoming sets associated to the characteristic flow of $Y$ and $d\nu^\pm$ are associated measures. In our case, we can split $D_T^+\approx\gamma_T^+\cup\left(\left\{T\right\}\times\Omega\times\R^3\right)$, $D_T^-\approx\gamma_T^-\cup\left(\left\{0\right\}\times\Omega\times\R^3\right)$ up to negligible sets (cf. \cite{BP87}). Then, $d\nu^\pm=d\gamm$ on $\gamma_T^\pm$ and $d\nu^\pm=dvdx$ on $t=0$ and $t=T$, and we decompose $f^+=\left(f_+,f\left(T\right)\right)$, $f^-=\left(f_-,f\left(0\right)\right)$ accordingly.
	
	\begin{proposition}\label{prop:Vleqest}
		Consider $\mathcal K=aK$, where $a\in L^\infty\left(\gamma_{\T}^-\right)$ such that $a_0:=\left\|a\right\|_{L^\infty\left(\gamma_{\T}^-\right)}<1$. Let $F$ be Lipschitz continuous, bounded, and divergence free with respect to $v$, and let $\mathring f\in\left(L^1\cap L^\infty\right)\left(\Omega\times\R^3\right)$, $g\in \left(L_\lt^1\cap L_\lt^\infty\right)\left(\gamma_{\T}^-,d\gamm\right)$ both be nonnegative. Then there is a unique, nonnegative strong solution $f\in L_{\lt}^\infty\left(\I;\left(L^1\cap L^\infty\right)\left(\Omega\times\R^3\right)\right)$ with nonnegative trace $f_\pm\in \left(L_\lt^1\cap L_\lt^\infty\right)\left(\gamma_{\T}^\pm,d\gamm\right)$ of \eqref{eq:linearVlasov} on $\I$. In particular, Definition \ref{def:WeakSolWholeSys} (ii) holds for $\left(f,f_+\right)$, where the Lorentz force is replaced by $F$. Moreover, we have
		\begin{align}\label{eq:Vleqest1}
		\left(1-a_0\right)^{\frac{1}{p}}\left\|f_+\right\|_{L^p\left(\gamma_T^+,d\gamm\right)},\left\|f\left(T\right)\right\|_{L^p\left(\Omega\times\R^3\right)}\leq\left\|\mathring f\right\|_{L^p\left(\Omega\times\R^3\right)}+\left(1-a_0\right)^{\frac{1}{p}-1}\left\|g\right\|_{L^p\left(\gamma_T^-,d\gamm\right)}
		\end{align}
		for any $0<T\in\I$ and $1\leq p\leq\infty$. If additionally $\mathring f\in L_{\kin}^1\left(\Omega\times\R^3\right)$ and $g\in L_{\kin,\lt}^1\left(\gamma_{\T}^-,d\gamm\right)$, then
		\begin{align}\label{eq:Vleqest3}
		&\left(1-a_0\right)\int_{\gamma_T^+\cap\left\{\left|v\right|<R\right\}}\vo f_+\,d\gamm+\int_{\Omega}\int_{B_R}\vo f\left(T\right)\,dvdx\nonumber\\
		&\leq\int_{\Omega}\int_{\R^3}\vo\mathring f\,dvdx+\int_{\gamma_T^-}\vo g\,d\gamm+\int_0^T\int_\Omega\int_{B_R}F\cdot\v f\,dvdxdt
		\end{align}
		and 
		\begin{align}\label{eq:Vleqest4}
		&\left\|\int_{B_R}f\left(T,\cdot,v\right)\,dv\right\|_{L^{\frac{4}{3}}\left(\Omega\right)}\nonumber\\
		&\leq\left(\frac{4\pi}{3}\left\|\mathring f\right\|_{L^\infty\left(\Omega\times\R^3\right)}+\frac{4\pi}{3}\left(1-a_0\right)^{-1}\left\|g\right\|_{L^\infty\left(\gamma_T^-\right)}+1\right)\left(\int_{\Omega}\int_{B_R}\vo f\left(T\right)\,dvdx\right)^{\frac{3}{4}}
		\end{align}
		for any $0<T\in\I$ and $0<R<\infty$.
	\end{proposition}
	\begin{proof}
		By \cite[Theorem 1]{BP87}, there is a unique, strong solution of \eqref{eq:linearVlasov} for each $T\in\I$. Since $T$ is arbitrary, we get $f\in L_\lt^p\left(\I\times\Omega\times\R^3\right)$ and $f_\pm\in L_\lt^p\left(\gamma_{\T}^\pm,d\gamm\right)$ for all $1\leq p<\infty$. By \cite[Proposition 1]{BP87}, we have the following $p$-norm estimate for $T\in\I$:
		\begin{align*}
		&\int_{\gamma_T^+}f_+^p\,d\gamm+\int_\Omega\int_{\R^3}f\left(T\right)^p\,dvdx\leq\int_{\Omega\times\R^3}\mathring f^p\,dvdx+\int_{\gamma_T^-}\left(aKf_++g\right)^p\,d\gamm\\
		&\leq\int_{\Omega\times\R^3}\mathring f^p\,dvdx+a_0\int_{\gamma_T^+}f_+^p\,d\gamm+\left(1-a_0\right)^{1-p}\int_{\gamma_T^-}g^p\,d\gamm
		\end{align*}
		using the convexity of the $p$-th power. This yields
		\begin{align*}
		\left(1-a_0\right)\int_{\gamma_T^+}f_+^p\,d\gamm+\int_\Omega\int_{\R^3}f\left(T\right)^p\,dvdx\leq\int_\Omega\int_{\R^3}\mathring f^p\,dvdx+\left(1-a_0\right)^{1-p}\int_{\gamma_T^-}g^p\,d\gamm
		\end{align*}
		and therefore \eqref{eq:Vleqest1} for $1\leq p<\infty$. Letting $p\to\infty$ we deduce \eqref{eq:Vleqest1} also for $p=\infty$. For this, note that $n\left(x\right)\cdot\v\gtrless 0$ on $\tilde\gamma^\pm$ which is why $L^\infty\left(\gamma_{\T}^\pm\right)=L^\infty\left(\gamma_{\T}^\pm,d\gamm\right)$ and the respective norms coincide.
		
		To prove the second estimate, let
		\begin{align*}
		\beta\colon\R^3\to\R,\quad\beta\left(v\right)=\begin{cases}\vo,&\left|v\right|<R,\\\sqrt{m_\alpha^2+R^2},&\left|v\right|\geq R.\end{cases}
		\end{align*}
		Noticing that $Y\left(\beta f\right)=F\cdot\beta'f$ and using the $1$-norm balance of \cite[Proposition 1]{BP87} we get by $\beta\geq 0$:
		\begin{align*}
		&\int_{\gamma_T^+}\beta f_+\,d\gamm+\int_\Omega\int_{\R^3}\beta f\left(T\right)\,dvdx\\
		&\leq\int_{\Omega}\int_{\R^3}\beta\mathring f\,dvdx+\int_{\gamma_T^-}\beta\left(aKf_++g\right)\,d\gamm+\int_0^T\int_\Omega\int_{\R^3}F\cdot\beta'f\,dvdxdt\\
		&\leq\int_{\Omega}\int_{\R^3}\beta\mathring f\,dvdx+a_0\int_{\gamma_T^+}\beta f_+\,d\gamm+\int_{\gamma_T^-}\beta g\,d\gamm+\int_0^T\int_\Omega\int_{\R^3}F\cdot\beta'f\,dvdxdt
		\end{align*}
		Writing the terms explicitly and using the fact that $\vo$ is monotonically increasing in $\left|v\right|$, we arrive at \eqref{eq:Vleqest3}.
		
		For \eqref{eq:Vleqest4}, we have
		\begin{align}\label{eq:jest43}
		\int_{B_R}f\,dv&\leq\int_{B_r}f\,dv+\int_{r\leq \left|v\right|<R}f\,dv\leq\frac{4\pi}{3}r^3\left\|f\left(T\right)\right\|_{L^\infty\left(\Omega\times\R^3\right)}+\frac{1}{r}\int_{B_R}\vo f\,dvdx\nonumber\\
		&\leq\left(\int_{B_R}\vo f\,dv\right)^{\frac{3}{4}}\left(\frac{4\pi}{3}\left\|\mathring f\right\|_{L^\infty\left(\Omega\times\R^3\right)}+\frac{4\pi}{3}\left(1-a_0\right)^{-1}\left\|g\right\|_{L^\infty\left(\gamma_T^-\right)}+1\right),
		\end{align}
		where we optimize $r:=\left(\int_{B_R}\vo f\,dv\right)^{\frac{1}{4}}$ in the standard manner. This yields \eqref{eq:Vleqest4}.
	\end{proof}
	Regarding the linear Maxwell part
	\begin{subequations}\label{eq:MaxwellPart}
		\begin{align}
		\varepsilon\partial_tE-\curl_xH&=-4\pi j,\\
		\mu\partial_tH+\curl_xE&=0,\\
		\left(E,H\right)\left(0\right)&=\left(\mathring E,\mathring H\right),
		\end{align}
	\end{subequations}
	on $\I$, there holds the following basic result:
	\begin{proposition}\label{prop:Maxwell}
		Let $\varepsilon,\mu\in H^3_{\mathrm{ul}}\left(\R^3;\R^{3\times 3}\right)$ have the following properties: $\varepsilon\left(x\right)$, $\mu\left(x\right)$ are symmetric for each $x\in\R^3$ and there is a $\sigma>0$ such that $\varepsilon\left(x\right),\mu\left(x\right)\geq\sigma$ for all $x\in\R^3$. Moreover let $j\in L_\lt^1\left(\I;H^3\left(\R^3;\R^3\right)\right)\cap C_\lt\left(\I;H^2\left(\R^3;\R^3\right)\right)$ and $\mathring E,\mathring H\in H^3\left(\R^3;\R^3\right)$. Then there is a unique solution $\left(E,H\right)\in C_\lt\left(\I;H^3\left(\R^3;\R^6\right)\right)\cap C_\lt^1\left(\I;H^2\left(\R^3;\R^6\right)\right)$ of \eqref{eq:MaxwellPart}. Furthermore we have
		\begin{align}\label{eq:Maxwelliden}
		\frac{1}{8\pi}\int_{\R^3}\left(\varepsilon E\cdot E+\mu H\cdot H\right)\left(T\right)\,dx=\frac{1}{8\pi}\int_{\R^3}\left(\varepsilon\mathring E\cdot\mathring E+\mu\mathring H\cdot\mathring H\right)\,dx-\int_0^T\int_{\R^3}E\cdot j\,dxdt
		\end{align} 
		and
		\begin{align}\label{eq:Maxwellest}
		\left\|\left(E,H\right)\left(T\right)\right\|_{L^2\left(\R^3;\R^6\right)}&:=\left(\left\|E\left(T\right)\right\|_{L^2\left(\R^3;\R^3\right)}^2+\left\|H\left(T\right)\right\|_{L^2\left(\R^3;\R^3\right)}^2\right)^{\frac{1}{2}}\nonumber\\
		&\leq\sigma^{-\frac{1}{2}}\left(\int_{\R^3}\left(\varepsilon\mathring E\cdot\mathring E+\mu\mathring H\cdot\mathring H\right)\,dx\right)^{\frac{1}{2}}+4\pi\sigma^{-1}\left\|j\right\|_{L^1\left(0,T;L^2\left(\R^3;\R^3\right)\right)}
		\end{align}
		for any $0<T\in\I$.
	\end{proposition}
	\begin{proof}
		For the existence theory (and a definition of uniform local Sobolev spaces $H_{\mathrm{ul}}^k$) we refer to \cite{Kat75}. Equation \eqref{eq:Maxwelliden} is derived straightforwardly by differentiating both sides and using the symmetry of $\varepsilon$ and $\mu$. We then get \eqref{eq:Maxwellest} by applying Lemma \ref{lma:QuadraticGronwall} using the uniform positive definiteness of $\varepsilon$ and $\mu$.
	\end{proof}
	Here and later, we need the following version of the quadratic Gronwall lemma, which is a slight improvement of \cite[Theorem 5]{Dra03}:
	\begin{lemma}\label{lma:QuadraticGronwall}
		Let $a,b\in\R$, $a<b$, $x,u\colon\left[a,b\right]\to\left[0,\infty\right[$ be continuous, $g\colon\left[a,b\right]\to\R$ be differentiable, and $\overline x\colon\left[a,b\right]\to\R$. Assume that the following inequality holds for all $t\in\left[a,b\right]$:
		\begin{align*}
		\frac{1}{2}\overline x\left(t\right)^2+\frac{1}{2}x\left(t\right)^2\leq\frac{1}{2}g\left(t\right)^2+\int_a^tu\left(s\right)x\left(s\right)\,ds.
		\end{align*}
		Then we have
		\begin{align*}
		\sqrt{\overline x\left(t\right)^2+x\left(t\right)^2}\leq\left|g\left(t\right)\right|+\int_a^tu\left(s\right)\,ds
		\end{align*}
		for all $t\in\left[a,b\right]$.
	\end{lemma}
	\begin{proof}
		Let $\iota>0$ and consider
		\begin{align*}
		y_\iota\colon\left[a,b\right]\to\left]0,\infty\right[,y_\iota\left(t\right)=\frac{1}{2}\left(g\left(t\right)^2+\iota^2\right)+\int_a^tu\left(s\right)x\left(s\right)\,ds.
		\end{align*}
		By assumption we have $x\left(t\right)\leq\sqrt{\overline x\left(t\right)^2+x\left(t\right)^2}\leq\sqrt{2y_\iota\left(t\right)}$. Furthermore, $\sqrt{2y_\iota}$ is differentiable with
		\begin{align*}
		\frac{d}{dt}\sqrt{2y_\iota\left(t\right)}=\frac{g\left(t\right)g'\left(t\right)+u\left(t\right)x\left(t\right)}{\sqrt{2y_\iota\left(t\right)}}\leq\frac{g\left(t\right)g'\left(t\right)}{\sqrt{g\left(t\right)^2+\iota^2}}+u\left(t\right).
		\end{align*}
		Integrating this estimate from $a$ to $t$ yields
		\begin{align*}
		&\sqrt{\overline x\left(t\right)^2+x\left(t\right)^2}\leq\sqrt{2y_\iota\left(t\right)}\leq\sqrt{2y_\iota\left(a\right)}+\int_a^t\frac{g\left(s\right)g'\left(s\right)}{\sqrt{g\left(s\right)^2+\iota^2}}\,ds+\int_a^tu\left(s\right)\,ds\\
		&=\sqrt{g\left(a\right)^2+\iota^2}+\sqrt{g\left(t\right)^2+\iota^2}-\sqrt{g\left(a\right)^2+\iota^2}+\int_a^tu\left(s\right)\,ds\leq\left|g\left(t\right)\right|+\iota+\int_a^tu\left(s\right)\,ds.
		\end{align*}
		Since $\iota>0$ is arbitrary, the proof is finished.
	\end{proof}
	
	\subsection{Approximations of the data}\label{sec:approx}
	Throughout this section we assume that Condition \ref{cond:data} is satisfied. We have to modify the data as follows to be able to apply the statements of Section \ref{sec:linearVM}: For $\alpha=1,\dots,N$ we define $\a_k:=\a$ and for $\alpha=N'+1,\dots,N$ we define $\a_k:=\frac{k}{k+1}\a$. Hence all $\a_k$ are bounded away from $1$. Furthermore, choose approximating sequences $\left(\mathring E_k\right)$, $\left(\mathring H_k\right)\subset H^3\left(\R^3;\R^3\right)$ with $\mathring E_k\to\mathring E$, $\mathring H_k\to\mathring H$ in $L^2\left(\R^3;\R^3\right)$ for $k\to\infty$. Additionally, we have to smooth $\varepsilon$ and $\mu$. In the following, have in mind that for a symmetric, positive definite matrix $A\in\R^{3\times 3}$ and some $C\geq 0$ we have the equivalence
	\begin{align*}
	A\leq C\Leftrightarrow\left\|A\right\|_{\R^{3\times 3}}\leq C
	\end{align*}
	where we use the norm
	\begin{align*}
	\left\|A\right\|_{\R^{3\times 3}}=\sup_{\left|x\right|\leq 1}\left|Ax\right|=\max\left\{\lambda\in\R\mid\lambda\mathrm{\ eigenvalue\ of\ }A\right\}
	\end{align*}
	where the last equality holds for symmetric, positive definite $A$. Thus, for some measurable $A\colon\R^3\to\R^{3\times 3}$ such that $A\left(x\right)$ is symmetric and positive definite for almost all $x\in\R^3$, the property $A\left(x\right)\leq C$ for almost all $x\in\R^3$ is equivalent to $\left\|A\right\|_{L^\infty\left(\R^3;\R^{3\times 3}\right)}\leq C$.
	
	We want to construct sequences of smooth $\varepsilon_k,\mu_k$ with $\sigma\leq\varepsilon_k,\mu_k\leq\sigma'$ in such a way that these sequences converge to $\varepsilon$, $\mu$ in a certain sense. We perform the construction of $\left(\varepsilon_k\right)$, the one for $\left(\mu_k\right)$ works totally analogously. Let $\omega\in C_c^\infty\left(\R^3\right)$, $\omega\geq 0$, $\supp\omega\subset\overline{B_1}$, $\int_{\R^3}\omega\,dx=1$ be a Friedrich's mollifier and define $\omega_s:=s^{-3}\omega\left(\frac{\cdot}{s}\right)$ for $s>0$. Now let
	\begin{align*}
	\tilde\varepsilon_k\left(x\right):=\begin{cases}\varepsilon\left(x\right)-\sigma I_3,&x\in B_k,\\0,&x\notin B_k\end{cases}
	\end{align*}
	for $k\in\N$, $I_3$ denoting the $3\times 3$-identity matrix. Clearly, $\tilde\varepsilon_k\in L^\infty\left(\R^3;\R^{3\times 3}\right)$ and $\tilde\varepsilon_k$ vanishes on $\R^3\setminus B_k$. This implies $\omega_s\ast\tilde\varepsilon_k\in C_c^\infty\left(\R^3;\R^{3\times 3}\right)$ (the convolution understood component-wise) for any $s>0$. By $\tilde\varepsilon_k\in L^2\left(B_k;\R^{3\times 3}\right)$ we know $\omega_s\ast\tilde\varepsilon_k\to\tilde\varepsilon_k$ in $L^2\left(B_k;\R^{3\times 3}\right)$ for $s\to 0$. Hence we can choose $s_k>0$ such that 
	\begin{align*}
	\left\|\omega_{s_k}\ast\tilde\varepsilon_k-\tilde\varepsilon_k\right\|_{L^2\left(B_k;\R^{3\times 3}\right)}<\frac{1}{k}.
	\end{align*}
	Finally define $\varepsilon_k:=\omega_{s_k}\ast\tilde\varepsilon_k+\sigma I_3$. Note that $\varepsilon_k$ is smooth and constant for $\left|x\right|$ large (and hence of class $H_{\mathrm{ul}}^3$). By construction, $\varepsilon_k\left(x\right)$ is symmetric for all $x\in\R^3$ and
	\begin{align}\label{eq:ApproxVareps}
	\left\|\varepsilon-\varepsilon_k\right\|_{L^2\left(B_k;\R^{3\times 3}\right)}<\frac{1}{k}.
	\end{align}
	Furthermore, for any $E,x\in\R^3$ it holds that
	\begin{align*}
	\varepsilon_k\left(x\right)E\cdot E&=\int_{\R^3}\omega_{s_k}\left(x-y\right)\tilde\varepsilon_k\left(y\right)E\cdot E\,dy+\sigma\left|E\right|^2\\
	&=\int_{B_k}\omega_{s_k}\left(x-y\right)\varepsilon\left(y\right)E\cdot E\,dy-\sigma\left|E\right|^2\int_{B_k}\omega_{s_k}\left(x-y\right)\,dy+\sigma\left|E\right|^2\\
	&\begin{cases}\displaystyle\geq\sigma\left|E\right|^2\int_{B_k}\omega_{s_k}\left(x-y\right)\,dy-\sigma\left|E\right|^2\int_{B_k}\omega_{s_k}\left(x-y\right)\,dy+\sigma\left|E\right|^2=\sigma\left|E\right|^2,\\\displaystyle\leq\sigma'\left|E\right|^2\int_{B_k}\omega_{s_k}\left(x-y\right)\,dy-\sigma\left|E\right|^2\int_{B_k}\omega_{s_k}\left(x-y\right)\,dy+\sigma\left|E\right|^2\leq\sigma'\left|E\right|^2.\end{cases}
	\end{align*}
	Note that for the last line we used the fact that the integral of $\omega_s$ over whole $\R^3$ equals $1$ for any $s>0$.
	\subsection{A cut-off problem}\label{sec:cutoff}
	In order to construct a weak solution of \eqref{eq:WholeSystem}, we first turn to a cut-off problem where we consider bounded time and momentum domains. Whereas the cut-off in time is no real drawback, the cut-off in momentum space is on the one hand unpleasant, but on the other hand necessary. To understand this necessity, we should recall \eqref{eq:Maxwellest}. Consider there $j$ to be the sum of some external current and the current $j^\inte$ induced by the particle densities. In an iteration scheme we would like to have an estimate like \eqref{eq:Maxwellest} for the fields where the right hand side is uniformly bounded along the iteration. Then we could extract some weakly converging subsequence. However, for this uniformity, we would need that $j^\inte$ is uniformly bounded in $L^1\left(0,T;L^2\left(\R^3;\R^3\right)\right)$ along the iteration. This would require a better estimate than \eqref{eq:Vleqest4} where we only can put our hands on the $L^{\frac{4}{3}}\left(\R^3;\R^3\right)$-norm of $j^\inte$ (at each time). Moreover, in an energy balance along the iteration, the crucial terms describing the energy transfer due to the internal system will not cancel out; this would only be the case if we solve $\eqref{eq:WholeSystem}$ simultaneously along an iteration.
	
	Now if we consider a cut-off problem (the cut-off referring to momentum space) we can simply estimate the $L^1$-norm of $j^\inte$ by the $L^2$-norm in momentum space and then use \eqref{eq:Vleqest1} for $p=2$, so we get the desired uniform boundedness along the iteration. Later, adding the limit versions of \eqref{eq:Vleqest3} and \eqref{eq:Maxwelliden}, we observe that the problematic terms on the right hand side, that is to say the terms $\pm E\cdot j^\inte$, cancel out. Thus, now (after a Gronwall argument) having a full energy estimate with only expressions of the data on the right hand side, we find that a posteriori the cut-off does not substantially enter this estimate, so we will be able to get a solution of the system without a cut-off by considering a sequence of solutions due to larger and larger cut-off domains.
	
	We differ from \cite{Guo93} as follows: Firstly, we do not have to cut off $\Omega$, since we only consider a bounded $\Omega$. Secondly, we solve the linear Vlasov equation on whole momentum space $\R^3$ and not only on a cut-off domain. Our cut-off only appears in the definition of the internal current $j_k^\inte$. Thirdly, as already said in the introduction, there is no need of the factor $e^{-t}$, and without this factor the estimates are more \enquote{natural}.
	
	To make things more precise, let $0<R<\infty$, define $R^*:=\min\left\{R,\T\right\}$, and start the iteration with $E_0,H_0\colon\left[0,R^*\right]\times\R^3\to\R^3$, $\left(E_0,H_0\right)\left(t,x,v\right)=\left(\mathring E_0,\mathring H_0\right)\left(x,v\right)$. We assume that we already have iterates of the $k$-th satisfying $E_k,H_k\in L^\infty\left(0,R^*;L^2\left(\R^3;\R^3\right)\right)\cap C^{0,1}\left(\left[0,R^*\right]\times\overline\Omega;\R^3\right)$.
	We first solve the Vlasov part
	\begin{subequations}\label{eq:VlasovIter}
		\begin{align}
		\partial_t\f_{k+1}+\v\cdot\partial_x\f_{k+1}+F_k^\alpha\cdot\partial_v\f_{k+1}&=0& \mathrm{on}\ \left[0,R^*\right]\times\Omega\times\R^3,\\
		\f_{k+1,-}&=\a_{k+1}K\f_{k+1,+}+\g&\mathrm{on}\ \gamma_{R^*}^-,\\
		\f_{k+1}\left(0\right)&=\mathring\f&\mathrm{on}\ \Omega\times\R^3
		\end{align}
	\end{subequations}
	with given force field $F_k^\alpha:=\e\left(E_k+\v\times H_k\right)$, which is Lipschitz continuous and bounded on $\left[0,R^*\right]\times\Omega\times\R^3$, and divergence free with respect to $v$. Indeed, we can solve \eqref{eq:VlasovIter} applying Proposition \ref{prop:Vleqest} and noticing that $\a_{k+1}$ is bounded away from $1$ on $\gamma_{R^*}^-$. Therefore we have $0\leq\f_{k+1}\in L^\infty\left(0,R^*;\left(L_{\kin}^1\cap L^\infty\right)\left(\Omega\times\R^3\right)\right)$ and $0\leq\f_{k+1,\pm}\in\left(L_{\kin}^1\cap L^\infty\right)\left(\gamma_{R^*}^\pm,d\gamm\right)$.
	
	Next we want to solve the Maxwell part. Now the cut-off appears: We define the current
	\begin{align}\label{eq:TotCurrIter}
	j_{k+1}:=j_{k+1}^\inte+u:=\suma\e\int_{B_R}\v\f_{k+1}\,dv+u
	\end{align}
	where we integrate only over the cut-off domain $B_R$ rather than over the whole momentum space. Note that $j_{k+1}^\inte$ ($u$) is defined to be $0$ outside $\Omega$ ($\Gamma$). By
	\begin{align}\label{eq:jintest}
	\left(\int_{\Omega}\left|j_{k+1}^\inte\right|^2\,dx\right)^{\frac{1}{2}}\leq\sqrt{\frac{4\pi}{3}R^3}\suma\left|\e\right|\left(\int_{\Omega}\int_{\R^3}\left|\f_{k+1}\right|^2\,dvdx\right)^{\frac{1}{2}}
	\end{align}
	and $\f_{k+1}\in L^\infty\left(0,R^*;L^2\left(\Omega\times\R^3\right)\right)$ we have $j_{k+1}\in L^1\left(0,R^*;L^2\left(\R^3\right)\right)$. In order to apply Proposition \ref{prop:Maxwell}, we approximate $j_{k+1}$ by a $\overline j_{k+1}\in C_c^\infty\left(\left]0,R^*\right[\times\R^3\right)$ such that 
	\begin{align}\label{eq:approxjk}
	4\pi\left\|j_{k+1}-\overline j_{k+1}\right\|_{L^1\left(0,R^*;L^2\left(\R^3;\R^3\right)\right)}<\frac{1}{k+1}.
	\end{align}
	With this smoothed current as the source term in the Maxwell system we solve
	\begin{subequations}\label{eq:MaxwellIter}
		\begin{align}
		\varepsilon_{k+1}\partial_tE_{k+1}-\curl_xH_{k+1}&=-4\pi\overline j_{k+1}&\mathrm{on}\ \left[0,R^*\right]\times\R^3,\\
		\mu_{k+1}\partial_tH_{k+1}+\curl_xE_{k+1}&=0&\mathrm{on}\ \left[0,R^*\right]\times\R^3,\\
		\left(E_{k+1},H_{k+1}\right)\left(0\right)&=\left(\mathring E_{k+1},\mathring H_{k+1}\right)&\mathrm{on}\ \R^3.
		\end{align}
	\end{subequations}
	Indeed, applying Proposition \ref{prop:Maxwell}, we see that there is a unique solution $\left(E_{k+1},H_{k+1}\right)\in C\left(0,R^*;H^3\left(\R^3;\R^6\right)\right)\cap C^1\left(0,R^*;H^2\left(\R^3;\R^6\right)\right)$. By Sobolev's embedding theorems it holds that $E_{k+1},H_{k+1}\in C^{0,1}\left(\left[0,R^*\right]\times\overline\Omega;\R^3\right)$. Altogether, the induction hypothesis is satisfied so that we can proceed with the next iteration step.
	
	In order to extract some weakly converging subsequence, we have to establish suitable estimates. To this end, consider \eqref{eq:Vleqest1} and \eqref{eq:Maxwellest} applied to \eqref{eq:VlasovIter} and \eqref{eq:MaxwellIter}:
	\begin{align}\label{eq:Vliterest}
	&\left(1-\left\|\a_{k+1}\right\|_{L^\infty\left(\gamma_{\T}^-\right)}\right)^{\frac{1}{p}}\left\|\f_{k+1,+}\right\|_{L^p\left(\gamma_T^+,d\gamm\right)},\left\|\f_{k+1}\left(T\right)\right\|_{L^p\left(\Omega\times\R^3\right)}\nonumber\\
	&\leq\left\|\mathring \f\right\|_{L^p\left(\Omega\times\R^3\right)}+\left(1-\left\|\a_{k+1}\right\|_{L^\infty\left(\gamma_{\T}^-\right)}\right)^{\frac{1}{p}-1}\left\|\g\right\|_{L^p\left(\gamma_T^-,d\gamm\right)}
	\end{align}
	and
	\begin{align}\label{eq:Maxwelliterest}
	\left\|\left(E_{k+1},H_{k+1}\right)\left(T\right)\right\|_{L^2\left(\R^3;\R^6\right)}&\leq\sigma^{-\frac{1}{2}}\left(\int_{\R^3}\left(\varepsilon_{k+1}\mathring E_{k+1}\cdot\mathring E_{k+1}+\mu_{k+1}\mathring H_{k+1}\cdot\mathring H_{k+1}\right)\,dx\right)^{\frac{1}{2}}\nonumber\\
	&\phantom{\leq\;}+4\pi\sigma^{-1}\left\|\overline j_{k+1}\right\|_{L^1\left(0,T;L^2\left(\R^3;\R^3\right)\right)}.
	\end{align}
	Note that we need $\varepsilon_k\left(x\right),\mu_k\left(x\right)\geq\sigma$ uniformly in $x$ and $k$ to get \eqref{eq:Maxwelliterest}.
	
	For $\alpha=1,\dots,N'$, \eqref{eq:Vliterest} reduces to
	\begin{align}\label{eq:Vliterest1}
	\left(1-\a_0\right)^{\frac{1}{p}}\left\|\f_{k+1,+}\right\|_{L^p\left(\gamma_T^+,d\gamm\right)},\left\|\f_{k+1}\left(T\right)\right\|_{L^p\left(\Omega\times\R^3\right)}\leq\left\|\mathring \f\right\|_{L^p\left(\Omega\times\R^3\right)}+\left(1-\a_0\right)^{\frac{1}{p}-1}\left\|\g\right\|_{L^p\left(\gamma_T^-,d\gamm\right)}
	\end{align}
	and to
	\begin{align}\label{eq:Vliterest2}
	\left(k+1\right)^{-\frac{1}{p}}\left\|\f_{k+1,+}\right\|_{L^p\left(\gamma_T^+,d\gamm\right)},\left\|\f_{k+1}\left(T\right)\right\|_{L^p\left(\Omega\times\R^3\right)}\leq\left\|\mathring \f\right\|_{L^p\left(\Omega\times\R^3\right)}
	\end{align}
	for $\alpha=N'+1,\dots,N$. Thus we conclude that any sequence $\left(\f_k\right)$ is bounded in any $L^p\left(\left[0,R^*\right]\times\Omega\times\R^3\right)$, $1\leq p\leq\infty$, so that we may extract a subsequence (also denoted by $\left(\f_k\right)$) that converges weakly in $L^p\left(\left[0,R^*\right]\times\Omega\times\R^3\right)$ for $1<p<\infty$ and weakly-* in $L^\infty\left(\left[0,R^*\right]\times\Omega\times\R^3\right)$ to some nonnegative $\f_R$. As in \eqref{eq:TotCurrIter} we define
	\begin{align*}
	j_R:=j_R^\inte+u:=\suma\e\int_{B_R}\v\f_R\,dv+u.
	\end{align*}
	As for the boundary values, we have to distinct absorbing and reflecting boundary conditions. For $\alpha=1,\dots,N'$, \eqref{eq:Vliterest1} yields the boundedness of $\left(\f_{k,+}\right)$ in any $L^p\left(\gamma_{R^*}^+,d\gamm\right)$, $1\leq p\leq\infty$, so we may extract a subsequence that converges weakly in $L^p\left(\gamma_{R^*}^+,d\gamm\right)$ for $1<p<\infty$ and weakly-* in $L^\infty\left(\gamma_{R^*}^+,d\gamm\right)$ to some nonnegative $\f_{R,+}$. For $\alpha=N'+1,\dots,N$, \eqref{eq:Vliterest2} delivers a uniform estimate only for $p=\infty$ so here we may extract a subsequence that only converges weakly-* to some nonnegative $\f_{R,+}$ in $L^\infty\left(\gamma_{R^*}^+,d\gamm\right)$.
	
	Letting $k\to\infty$, we deduce for $1\leq p\leq\infty$
	\begin{align}
	\left\|\f_R\right\|_{L^\infty\left(0,T;L^p\left(\Omega\times\R^3\right)\right)}&\leq\left\|\mathring \f\right\|_{L^p\left(\Omega\times\R^3\right)}+\begin{cases}\left(1-\a_0\right)^{\frac{1}{p}-1}\left\|\g\right\|_{L^p\left(\gamma_T^-,d\gamm\right)},&\alpha\leq N'\\0,&\alpha>N'\end{cases}\label{eq:fRest}\\
	\left\|\f_{R,+}\right\|_{L^\infty\left(\gamma_T^+,d\gamm\right)}&\leq\left\|\mathring\f\right\|_{L^\infty\left(\Omega\times\R^3\right)}+\begin{cases}\left(1-\a_0\right)^{-1}\left\|\g\right\|_{L^\infty\left(\gamma_T^-,d\gamm\right)},&\alpha\leq N'\\0,&\alpha>N'\end{cases}
	\end{align}
	and for $\alpha=N'+1,\dots,N$ additionally
	\begin{align}\label{eq:f+Restp}
	\left\|\f_{R,+}\right\|_{L^p\left(\gamma_T^+,d\gamm\right)}\leq\left(1-\a_0\right)^{-\frac{1}{p}}\left\|\mathring \f\right\|_{L^p\left(\Omega\times\R^3\right)}+\left(1-\a_0\right)^{-1}\left\|\g\right\|_{L^p\left(\gamma_T^-,d\gamm\right)}.
	\end{align}
	
	Next we turn to an estimate on the electromagnetic fields. To examine \eqref{eq:Maxwelliterest} further, we insert the properties of $\overline j_{k+1}$ on the right hand side to get
	\begin{align*}
	&\left\|\overline j_{k+1}\right\|_{L^1\left(0,T;L^2\left(\R^3;\R^3\right)\right)}\leq\frac{1}{4\pi\left(k+1\right)}+\left\|j_{k+1}\right\|_{L^1\left(0,R^*;L^2\left(\R^3;\R^3\right)\right)}\\
	&\leq 1+\sqrt{\frac{4\pi}{3}R^3}\suma\left|\e\right|\int_0^{R^*}\left\|\f_{k+1}\left(t\right)\right\|_{L^2\left(\Omega\times\R^3\right)}\,dt+\left\|u\right\|_{L^1\left(0,R^*;L^2\left(\Gamma;\R^3\right)\right)}
	\end{align*}
	for $0<T\leq R^*$ using \eqref{eq:jintest}. The right hand side is bounded uniformly in $k$. Moreover, the first term on the right hand side of \eqref{eq:Maxwelliterest} is bounded uniformly in $k$ by $\varepsilon_k,\mu_k\leq\sigma'$ and the $L^2$-convergence of the approximating initial data. Thus, we may extract a subsequence $\left(E_k,H_k\right)$ that converges weakly in $L^2\left([0,R^*]\times\R^3;\R^6\right)$ to some $\left(E_R,H_R\right)$.
	
	We now show that $\left(\left(\f_R,\f_{R,+}\right)_\alpha,E_R,H_R,j_R\right)$ is a weak solution of \eqref{eq:WholeSystem} on the time interval $\left[0,R^*\right]$ in the sense of Definition \ref{def:WeakSolWholeSys}. Clearly, all functions are of class $L_\loc^1$. The main task is to show that we may pass to the limit in \eqref{eq:Vlasovweak} and \eqref{eq:Maxwellweak} applied to the iterates: We have for all $\psi\in\Psi_{R^*}$, $\vartheta\in\Theta_{R^*}$, and $k\geq 1$
	\begin{align}
	0&=-\int_0^{R^*}\int_\Omega\int_{\R^3}\left(\partial_t\psi+\v\cdot\partial_x\psi+\e\left(E_k+\v\times H_k\right)\cdot\partial_v\psi\right)\f_{k+1}\,dvdxdt\nonumber\\
	&\phantom{=\;}+\int_{\gamma_{R^*}^+}\f_{k+1,+}\psi\,d\gamm-\int_{\gamma_{R^*}^-}\left(\K\f_{k+1,+}+\g\right)\psi\,d\gamm-\int_\Omega\int_{\R^3}\mathring\f\psi\left(0\right)\,dvdx,\label{eq:Vlasovweakiter}\\
	0&=\int_0^{R^*}\int_{\R^3}\left(\varepsilon_kE_k\cdot\partial_t\vartheta-H_k\cdot\curl_x\vartheta-4\pi\overline j_k\cdot\vartheta\right)\,dxdt+\int_{\R^3}\varepsilon_k\mathring{E}_k\cdot\vartheta\left(0\right)\,dx,\label{eq:Maxwellweakiter1}\\
	0&=\int_0^{R^*}\int_{\R^3}\left(\mu_kH_k\cdot\partial_t\vartheta+E_k\cdot\curl_x\vartheta\right)\,dxdt+\int_{\R^3}\mu_k\mathring{H}_k\cdot\vartheta\left(0\right)\,dx.\label{eq:Maxwellweakiter2}
	\end{align}
	We can pass to the limit in \eqref{eq:Maxwellweakiter1} and \eqref{eq:Maxwellweakiter2}: Whereas the terms including the curl are easy to handle by weak convergence of $E_k$, $H_k$, we have to take more care about the terms including $\varepsilon_k,\mu_k$, and $\overline j_k$. For the first ones, let $K\in\N$ such that $\vartheta$ vanishes for $\left|x\right|\geq K$ so that we in fact only integrate over $B_K$. For $k\geq K$ we have
	\begin{align*}
	\left\|\varepsilon-\varepsilon_k\right\|_{L^2\left(B_K;\R^{3\times 3}\right)}\leq\left\|\varepsilon-\varepsilon_k\right\|_{L^2\left(B_k;\R^{3\times 3}\right)}<\frac{1}{k}
	\end{align*}
	by \eqref{eq:ApproxVareps} so that $\varepsilon_k\to\varepsilon$ in $L^2\left(B_K;\R^{3\times 3}\right)$. This is enough for passing to the limit in the terms including $\varepsilon_k$ since we additionally have $E_k\rightharpoonup E_R$ in $L^2\left([0,R^*]\times\R^3;\R^3\right)$, even strong convergence of the approximating initial data, and the boundedness of the time interval $\left[0,R^*\right]$. Similarly, we argue for the terms with $\mu_k$. So there only remains the term including $\overline j_k$. To tackle this one, we estimate
	\begin{align*}
	\left|\int_0^{R^*}\int_{\R^3}\left(\overline j_k-j_R\right)\cdot\vartheta\,dxdt\right|&\leq\left\|\overline j_k-j_k\right\|_{L^1\left(0,R^*;L^2\left(\R^3;\R^3\right)\right)}\left\|\vartheta\right\|_{L^\infty\left(0,R^*;L^2\left(\R^3;\R^3\right)\right)}\\
	&\phantom{\leq\;}+\suma\left|\e\right|\left|\int_0^{R^*}\int_{\R^3}\int_{B_R}\v\left(\f_k-\f_R\right)\,dv\cdot\vartheta\,dxdt\right|
	\end{align*}
	where the first term on the right hand side converges to $0$ for $k\to\infty$ by construction of $\overline j_k$ and each summand of the second term by weak convergence of the $\f_k$. Note that for the latter limit our cut-off plays an important role since $\v\cdot\vartheta\chi_{\left\{\left|v\right|\leq R\right\}}\in L^2\left(\left[0,R^*\right]\times\R^3\times\R^3\right)$.
	
	Passing to the limit in \eqref{eq:Vlasovweakiter} is more complicated, especially because of the nonlinear product term including $E_k$, $H_k$, and $\f_k$. The other terms are easy to handle due to weak convergence of $\f_k$ and weak (or weakly-*) convergence of $\f_{k,+}$. The nonlinear term is handled as in \cite[Proof of Lemma 3.1.]{Guo93} by a highly nontrivial tool, namely the momentum-averaging lemma (see \cite{DL89}, or \cite{Rei04} for a shortened proof). For this, it is important that the sequences $\left(\f_k\right)$ are bounded in the $L^2$- and $L^\infty$-norm and $\left(E_k,H_k\right)$ is bounded in the $L^2$-norm.
	
	Altogether, $\left(\left(\f_R,\f_{R,+}\right)_\alpha,E_R,H_R,j_R\right)$ is a weak solution of \eqref{eq:WholeSystem} on the time interval $\left[0,R^*\right]$ in the sense of Definition \ref{def:WeakSolWholeSys}.
	
	In order to have good estimates for $R\to\infty$, the right hand side of an energy inequality should not depend on $R$. To this end, consider \eqref{eq:Vleqest3} and \eqref{eq:Maxwelliden} applied to the $k$-iterated functions. Note that the estimate on the term on the left hand side of \eqref{eq:Vleqest3} including the boundary values is only worth anything for $k\to\infty$ for $\alpha=1,\dots,N'$. Therefore, it is convenient to introduce
	\begin{align*}
	\b_k\left(T\right):=\begin{cases}\displaystyle\left(1-\a_0\right)\int_{\gamma_T^+\cap\left\{\left|v\right|<R\right\}}\vo \f_{k,+}\,d\gamm,&\alpha=1,\dots,N'\\0,&\alpha=N'+1,\dots,N\end{cases}
	\end{align*}
	and similarly $\b_R\left(T\right)$ where $k$ is replaced by $R$. Now we have
	\begin{align}\label{eq:Vlenergyineqiter}
	&\b_k\left(T\right)+\int_\Omega\int_{B_R}\vo\f_k\left(T\right)\,dvdx\nonumber\\
	&\leq\int_\Omega\int_{\R^3}\vo\mathring\f\,dvdx+\int_{\gamma_T^-}\vo\g\,d\gamm+\int_0^T\int_\Omega\int_{B_R}\e\left(E_{k-1}+\v\times H_{k-1}\right)\cdot\v\f_k\,dvdxdt\nonumber\\
	&=\int_\Omega\int_{\R^3}\vo\mathring\f\,dvdx+\int_{\gamma_T^-}\vo\g\,d\gamm+\int_0^T\int_\Omega E_{k-1}\cdot\int_{B_R}\e\v\f_k\,dvdxdt
	\end{align}
	and
	\begin{align}\label{eq:Maxenergyiter}
	&\frac{1}{8\pi}\int_{\R^3}\left(\varepsilon_kE_k\cdot E_k+\mu_kH_k\cdot H_k\right)\left(T\right)\,dx\nonumber\\
	&=\frac{1}{8\pi}\int_{\R^3}\left(\varepsilon_k\mathring E_k\cdot\mathring E_k+\mu_k\mathring H_k\cdot\mathring H_k\right)\,dx-\int_0^T\int_{\R^3}E_k\cdot\overline j_k\,dxdt
	\end{align}
	for $k\geq 1$ and any $T\in \left]0,R^*\right]$. We consider the right hand sides of \eqref{eq:Vlenergyineqiter} and \eqref{eq:Maxenergyiter} further. The term including the initial data of the electromagnetic fields is bounded uniformly in $k$ due to
	\begin{align*}
	\int_{\R^3}\left(\varepsilon_k\mathring E_k\cdot\mathring E_k+\mu_k\mathring H_k\cdot\mathring H_k\right)\,dx\leq\sigma'\int_{\R^3}\left(\left|\mathring E_k\right|^2+\left|\mathring H_k\right|^2\right)\,dx\overset{k\to\infty}\to\sigma'\int_{\R^3}\left(\left|\mathring E\right|^2+\left|\mathring H\right|^2\right)\,dx.
	\end{align*}
	After approximating $\e\widehat{\cdot}_\alpha$ in $L^2\left(B_R;\R^3\right)$ by $C_c^\infty\left(B_R;\R^3\right)$-functions and using the momentum averaging lemma again we have, up to a subsequence,
	\begin{align}\label{eq:LimEnerTrans}
	\lim_{k\to\infty}\int_0^T\int_\Omega E_{k-1}\cdot\int_{B_R}\e\v\f_k\,dvdxdt=\int_0^T\int_\Omega E_R\cdot\int_{B_R}\e\v\f_R\,dvdxdt.
	\end{align}
	Summing \eqref{eq:LimEnerTrans} over $\alpha$ yields
	\begin{align*}
	\lim_{k\to\infty}\int_0^T\int_\Omega E_{k-1}\cdot j_k^\inte\,dvdxdt=\int_0^T\int_\Omega E_R\cdot j_R^\inte\,dvdxdt.
	\end{align*}
	Similarly,
	\begin{align*}
	\lim_{k\to\infty}\int_0^T\int_\Omega E_k\cdot j_k^\inte\,dvdxdt=\int_0^T\int_\Omega E_R\cdot j_R^\inte\,dvdxdt,
	\end{align*}
	whence we have
	\begin{align}\label{eq:EnergyTransIter}
	\lim_{k\to\infty}\int_0^T\int_\Omega\left(E_{k-1}\cdot j_k^\inte-E_k\cdot j_k^\inte\right)\,dvdxdt=0.
	\end{align}
	Unfortunately, this is not enough since we in fact have to consider $E_{k-1}\cdot j_k^\inte-E_k\cdot\overline j_k$. To get hands on this term, choose $\varphi_k^1,\varphi_k^2\in C_c^\infty\left(\left]0,R^*\right[\times\R^3\right)$ with
	\begin{align}\label{eq:approxEnerTransIter}
	\left\|E_{k-1}\cdot j_k^\inte-\varphi_k^1\right\|_{L^1\left(\left]0,R^*\right[\times\R^3\right)},\left\|E_k\cdot j_k^\inte-\varphi_k^2\right\|_{L^1\left(\left]0,R^*\right[\times\R^3\right)}<\frac{1}{k}
	\end{align}
	and choose $u_k\in C_c^\infty\left(\left]0,R^*\right[\times\Gamma;\R^3\right)$ such that
	\begin{align*}
	\left\|u-u_k\right\|_{L^1\left(0,R^*;L^2\left(\Gamma;\R^3\right)\right)}<\frac{1}{k}.
	\end{align*}
	Using these approximations and \eqref{eq:TotCurrIter} and \eqref{eq:approxjk} we estimate
	\begin{align}\label{eq:EnerTranIterEst}
	&\left|\int_0^T\int_{\R^3}\left(E_{k-1}\cdot j_k^\inte-E_k\cdot\overline j_k\right)\,dxdt\right|\nonumber\\
	&\leq\left|\int_0^T\int_{\R^3}E_k\cdot u_k\,dxdt\right|+\left|\int_0^T\int_{\R^3}E_k\cdot\left(u-u_k\right)\,dxdt\right|+\left|\int_0^T\int_{\R^3}\left(\varphi_k^1-\varphi_k^2\right)\,dxdt\right|\nonumber\\
	&\phantom{=\;}+\left|\int_0^T\int_{\R^3}\left(E_{k-1}\cdot j_k^\inte-\varphi_k^1\right)\,dxdt\right|+\left|\int_0^T\int_{\R^3}\left(\varphi_k^2-E_k\cdot j_k^\inte\right)\,dxdt\right|\nonumber\\
	&\phantom{=\;}+\left|\int_0^T\int_{\R^3}E_k\cdot\left(j_k-\overline j_k\right)\,dxdt\right|\nonumber\\
	&\leq\int_0^T\left\|E_k\left(t\right)\right\|_{L^2\left(\R^3;\R^3\right)}\left\|u_k\left(t\right)\right\|_{L^2\left(\Gamma;\R^3\right)}\,dt+\left|\int_0^T\int_{\R^3}\left(\varphi_k^1-\varphi_k^2\right)\,dxdt\right|+\frac{C}{k}\nonumber\\
	&=:\int_0^T\left\|E_k\left(t\right)\right\|_{L^2\left(\R^3;\R^3\right)}\left\|u_k\left(t\right)\right\|_{L^2\left(\Gamma;\R^3\right)}\,dt+h_k\left(T\right)
	\end{align}
	where $C>0$ does not depend on $k$ since we already have a uniform bound on the $E_k$ in $L^\infty\left(0,R^*;L^2\left(\R^3;\R^3\right)\right)$. Furthermore, $h_k$ is continuous with respect to $T$ and
	\begin{align*}
	h_k\left(T\right)\to 0\mathrm{\ for\ }k\to\infty\mathrm{\ for\ each\ }T\in\left[0,R^*\right]
	\end{align*}
	by \eqref{eq:EnergyTransIter} and \eqref{eq:approxEnerTransIter}. Moreover, we have
	\begin{align*}
	&\left|h_k\left(T\right)\right|\leq\frac{C+2}{k}+\left\|E_{k-1}\cdot j_k^\inte\right\|_{L^1\left(\left]0,R^*\right[\times\Omega\right)}+\left\|E_k\cdot j_k^\inte\right\|_{L^1\left(\left]0,R^*\right[\times\Omega\right)}\\
	&\leq\frac{C}{k}+\left(\left\|E_{k-1}\right\|_{L^\infty\left(0,R^*;L^2\left(\R^3;\R^3\right)\right)}+\left\|E_k\right\|_{L^\infty\left(0,R^*;L^2\left(\R^3;\R^3\right)\right)}\right)\left\|j_k^\inte\right\|_{L^1\left(0,R^*;L^2\left(\Omega;\R^3\right)\right)}\leq C
	\end{align*}
	where $C>0$ does not depend on $k$ (and $T$) by the uniform boundedness of the $E_k$ in $L^\infty\left(0,R^*;L^2\left(\R^3;\R^3\right)\right)$ and \eqref{eq:jintest} (combined with \eqref{eq:Vliterest1} and \eqref{eq:Vliterest2}, respectively). Therefore we can choose $l_k\in C^1\left(\left[0,R^*\right]\right)$ such that
	\begin{align}\label{eq:approxhk}
	\left\|\sqrt{h_k}-l_k\right\|_{C\left(\left[0,R^*\right]\right)}<\frac{1}{k}.
	\end{align}
	Then there also holds
	\begin{align}\label{eq:lklebesgue}
	l_k\left(T\right)\to 0\mathrm{\ for\ }k\to\infty\mathrm{\ and\ }\left|l_k\left(T\right)\right|\leq C+1\mathrm{\ for\ each\ }T\in\left[0,R^*\right].
	\end{align}
	Now let $0<T\leq T'\leq R^*$. Exploiting $\sigma\leq\varepsilon_k,\mu_k\leq\sigma'$, summing \eqref{eq:Vlenergyineqiter} over $\alpha$, adding \eqref{eq:Maxenergyiter}, and then using \eqref{eq:EnerTranIterEst} and \eqref{eq:approxhk} yields
	\begin{align*}
	&\suma\b_k\left(T\right)+\suma\int_\Omega\int_{B_R}\vo\f_k\left(T\right)\,dvdx+\frac{\sigma}{8\pi}\left\|\left(E_k,H_k\right)\left(T\right)\right\|_{L^2\left(\R^3;\R^6\right)}^2\\
	&\leq\suma\b_k\left(T\right)+\suma\int_\Omega\int_{B_R}\vo\f_k\left(T\right)\,dvdx+\frac{1}{8\pi}\int_{\R^3}\left(\varepsilon_kE_k\cdot E_k+\mu_kH_k\cdot H_k\right)\left(T\right)\,dx\\
	&\leq\suma\int_\Omega\int_{\R^3}\vo\mathring\f\,dvdx+\suma\int_{\gamma_T^-}\vo\g\,d\gamm+\frac{1}{8\pi}\int_{\R^3}\left(\varepsilon_k\mathring E_k\cdot\mathring E_k+\mu_k\mathring H_k\cdot\mathring H_k\right)\,dx\\
	&\phantom{=\;}+\int_0^T\int_\Omega\left(E_{k-1}\cdot j_k^\inte-E_k\cdot\overline j_k\right)\,dxdt\\
	&\leq\suma\int_\Omega\int_{\R^3}\vo\mathring\f\,dvdx+\suma\int_{\gamma_{T'}^-}\vo\g\,d\gamm+\frac{\sigma'}{8\pi}\left\|\left(\mathring E_k,\mathring H_k\right)\right\|_{L^2\left(\R^3;\R^6\right)}^2\\
	&\phantom{=\;}+\int_0^T\left\|E_k\left(t\right)\right\|_{L^2\left(\R^3;\R^3\right)}\left\|u_k\left(t\right)\right\|_{L^2\left(\Gamma;\R^3\right)}\,dt+h_k\left(T\right)\\
	&\leq\suma\int_\Omega\int_{\R^3}\vo\mathring\f\,dvdx+\suma\int_{\gamma_{T'}^-}\vo\g\,d\gamm+\frac{\sigma'}{8\pi}\left\|\left(\mathring E_k,\mathring H_k\right)\right\|_{L^2\left(\R^3;\R^6\right)}^2\\
	&\phantom{=\;}+\sqrt{4\pi}\sigma^{-\frac{1}{2}}\int_0^T\frac{\sqrt{\sigma}}{\sqrt{4\pi}}\left\|\left(E_k,H_k\right)\left(t\right)\right\|_{L^2\left(\R^3;\R^6\right)}\left\|u_k\left(t\right)\right\|_{L^2\left(\Gamma;\R^3\right)}\,dt+2l_k\left(T\right)^2+\frac{2}{k^2}.
	\end{align*}
	By $E_k,H_k\in C\left(0,R^*;L^2\left(\R^3;\R^3\right)\right)$, $u_k\in C\left(0,R^*;L^2\left(\Gamma;\R^3\right)\right)$, and by differentiability of $l_k$ we can apply Lemma \ref{lma:QuadraticGronwall} and thus obtain
	\begin{align}\label{eq:EnerIneqIter}
	&\suma\b_k\left(T\right)+\suma\int_\Omega\int_{B_R}\vo\f_k\left(T\right)\,dvdx+\frac{\sigma}{8\pi}\left\|\left(E_k,H_k\right)\left(T\right)\right\|_{L^2\left(\R^3;\R^6\right)}^2\nonumber\\
	&\leq\left(\left(\suma\int_\Omega\int_{\R^3}\vo\mathring\f\,dvdx+\suma\int_{\gamma_{T'}^-}\vo\g\,d\gamm+\frac{\sigma'}{8\pi}\left\|\left(\mathring E_k,\mathring H_k\right)\right\|_{L^2\left(\R^3;\R^6\right)}^2\right.\right.\nonumber\\
	&\omit\hfill$\displaystyle\left.\left.\vphantom{\suma\int_\Omega\int_{\R^3}\vo\mathring\f\,dvdx+\suma\int_{\gamma_{T'}^-}\vo\g\,d\gamm+\frac{\sigma'}{2}\left\|\left(\mathring E_k,\mathring H_k\right)\right\|_{L^2\left(\R^3;\R^6\right)}^2}+2l_k\left(T\right)^2+\frac{2}{k^2}\right)^{\frac{1}{2}}+\sqrt{2\pi}\sigma^{-\frac{1}{2}}\left\|u\right\|_{L^1\left(0,T';L^2\left(\Gamma;\R^3\right)\right)}+\sqrt{2\pi}\sigma^{-\frac{1}{2}}\frac{1}{k}\right)^2$
	\end{align}
	altogether. For $k\to\infty$, let $A\subset\left[0,T'\right]$ be measurable and integrate \eqref{eq:EnerIneqIter} over $A$. As for $\suma\b_k\left(T\right)$, we note that $\suma\b_R\left(T\right)$ is the pointwise limit of $\suma\b_k\left(T\right)$ by weak convergence and we have a pointwise bound uniformly in $T$ and $k$ by \eqref{eq:EnerIneqIter}. Additionally exploiting weak convergence and weak lower semi-continuity, respectively, the strong convergence of the initial electromagnetic fields, and \eqref{eq:lklebesgue} we may pass to the limit and conclude, since $A$ was arbitrary, that
	\begin{align}\label{eq:EnerEstR}
	&\left(\sumb\left(1-\a_0\right)\int_{\gamma_T^+\cap\left\{\left|v\right|<R\right\}}\vo \f_{R,+}\,d\gamm\vphantom{\left\|\suma\int_\Omega\int_{B_R}\vo\f_R\left(\cdot\right)\,dvdx+\frac{\sigma}{8\pi}\left\|\left(E_R,H_R\right)\left(\cdot\right)\right\|_{L^2\left(\R^3;\R^6\right)}^2\right\|_{L^\infty\left(\left[0,T\right]\right)}}\right.\nonumber\\
	&\omit\hfill$\displaystyle\left.+\left\|\suma\int_\Omega\int_{B_R}\vo\f_R\left(\cdot\right)\,dvdx+\frac{\sigma}{8\pi}\left\|\left(E_R,H_R\right)\left(\cdot\right)\right\|_{L^2\left(\R^3;\R^6\right)}^2\right\|_{L^\infty\left(\left[0,T\right]\right)}\right)^{\frac{1}{2}}$\nonumber\\
	&\leq\left(\suma\int_\Omega\int_{\R^3}\vo\mathring\f\,dvdx+\sumb\int_{\gamma_T^-}\vo\g\,d\gamm+\frac{\sigma'}{8\pi}\left\|\left(\mathring E,\mathring H\right)\right\|_{L^2\left(\R^3;\R^6\right)}^2\right)^{\frac{1}{2}}\nonumber\\
	&\phantom{=\;}+\sqrt{2\pi}\sigma^{-\frac{1}{2}}\left\|u\right\|_{L^1\left(0,T;L^2\left(\Gamma;\R^3\right)\right)}
	\end{align}
	for all $T\in\left]0,R^*\right]$, after taking $T=T'$. This is exactly the energy estimate we wanted to derive since $R$ does no longer appear on the right hand side.
	
	Lastly, we show that, up to a subsequence, $j_k^\inte\rightharpoonup j_R^\inte$ in $L^{\frac{4}{3}}\left(\left[0,R^*\right]\times\Omega;\R^3\right)$ for $k\to\infty$ and derive an $L^\infty\left(0,R^*;L^{\frac{4}{3}}\left(\Omega;\R^3\right)\right)$-bound for $j_R^\inte$. To this end, applying \eqref{eq:Vleqest4} yields
	\begin{align*}
	&\left\|j_k^\inte\left(T\right)\right\|_{L^{\frac{4}{3}}\left(\Omega;\R^3\right)}\leq\suma\left|\e\right|\left\|\int_{B_R}\f_k\left(T,\cdot,v\right)\,dv\right\|_{L^{\frac{4}{3}}\left(\Omega\right)}\\
	&\leq\suma\left(\frac{4\pi}{3}\left\|\mathring\f\right\|_{L^\infty\left(\Omega\times\R^3\right)}+1+\begin{cases}\frac{4\pi}{3}\left(1-\a_0\right)^{-1}\left\|\g\right\|_{L^\infty\left(\gamma_T^-\right)},&\alpha=1,\dots,N'\\0,&\alpha=N'+1,\dots,N\end{cases}\right)\\
	&\omit\hfill$\displaystyle\cdot\left|\e\right|\left(\int_{\Omega}\int_{B_R}\vo\f_k\left(T\right)\,dvdx\right)^{\frac{3}{4}}$
	\end{align*} 
	for $0\leq T\leq R^*$ and the right hand side is bounded in $L^{\frac{4}{3}}\left(\left[0,R^*\right]\right)$ uniformly in $k$ by virtue of \eqref{eq:EnerEstR}. Therefore we may assume that $j_k^\inte$ converges weakly in $L^{\frac{4}{3}}\left(\left[0,R^*\right]\times\Omega;\R^3\right)$. It is easy to see that the weak limit has to be $j_R^\inte$. As for the desired bound, we proceed similarly to \eqref{eq:Vleqest4} and \eqref{eq:jest43}, respectively, sum over $\alpha$, apply a Hölder estimate for the sum, and use the known estimates to get
	\begin{align}
	&\left\|j_R^\inte\right\|_{L^\infty\left(0,T;L^{\frac{4}{3}}\left(\Omega;\R^3\right)\right)}\nonumber\\
	&\leq\left(\suma\left|\e\right|^4\left(\frac{4\pi}{3}\left\|\mathring\f\right\|_{L^\infty\left(\Omega\times\R^3\right)}+1+\begin{cases}\frac{4\pi}{3\left(1-\a_0\right)}\left\|\g\right\|_{L^\infty\left(\gamma_T^-\right)},&\alpha\leq N'\\0,&\alpha>N'\end{cases}\right)^4\right)^{\frac{1}{4}}\nonumber\\
	&\phantom{=\;}\cdot\left(\left(\suma\int_\Omega\int_{\R^3}\vo\mathring\f\,dvdx+\sumb\int_{\gamma_T^-}\vo\g\,d\gamm+\frac{\sigma'}{8\pi}\left\|\left(\mathring E,\mathring H\right)\right\|_{L^2\left(\R^3;\R^6\right)}^2\right)^{\frac{1}{2}}\right.\nonumber\\
	&\omit\hfill$\displaystyle\left.\vphantom{\left(\suma\int_\Omega\int_{\R^3}\vo\mathring\f\,dvdx+\sumb\int_{\gamma_T^-}\vo\g\,d\gamm+\frac{\sigma'}{8\pi}\left\|\left(\mathring E,\mathring H\right)\right\|_{L^2\left(\R^3;\R^6\right)}^2\right)^{\frac{1}{2}}}+\sqrt{2\pi}\sigma^{-\frac{1}{2}}\left\|u\right\|_{L^1\left(0,T;L^2\left(\Gamma;\R^3\right)\right)}\right)^{\frac{3}{2}}$\label{eq:jRest}
	\end{align}
	for any $0<T\leq R^*$.
	
	\subsection{Removing the cut-off}\label{sec:remove}
	Finally we obtain a solution of \eqref{eq:WholeSystem} on the time Interval $\I$ by letting $R\to\infty$. To this end, it is crucial that the right hand sides of the obtained estimates of the previous section do not depend on $R$; see \eqref{eq:fRest} to \eqref{eq:f+Restp}, \eqref{eq:EnerEstR}, and \eqref{eq:jRest}. Take the sequence $\left(R_m\right)_m=\left(m\right)_m$, then we see by a diagonal sequence argument that, for certain limit functions, $\f_m\overset{(\ast)}\rightharpoonup\f$ in $L^p\left(\left[0,M^*\right]\times\Omega\times\R^3\right)$, $\f_{m,+}\overset{*}\rightharpoonup\f_+$ in $L^\infty\left(\left[0,M^*\right]\times\Omega\times\R^3\right)$, $\left(E_m,H_m\right)\rightharpoonup\left(E,H\right)$ in $L^2\left(\left[0,M^*\right]\times\R^3;\R^6\right)$, and $j_m^\inte\rightharpoonup j^\inte$ in $L^{\frac{4}{3}}\left(\left[0,M^*\right]\times\Omega;\R^3\right)$ for each $1<p\leq\infty$, $M>0$ (where $M^*=\min\left\{M,\T\right\}$). For $\alpha=1,\dots,N'$ we additionally have $\f_{m,+}\rightharpoonup\f_+$ in $L^p\left(\left[0,M^*\right]\times\Omega\times\R^3\right)$ for $1<p<\infty$. We may pass to the limit in the respective estimates to obtain \eqref{eq:estf1} to \eqref{eq:estjint}. Passage to the limit in the weak formulation of \eqref{eq:WholeSystem} works in the same way as in \cite[Theorem 4.1.]{Guo93}. That the weak limit of the $j_m^\inte$ is indeed the current density $j^\inte$ induced by the $\f$ is proved in the same way as in \cite[Proposition 4]{Rei04} exploiting the energy estimate.
	
	Altogether, Theorem \ref{thm:Existence} is proved.
	
	\section{The redundant divergence equations and the charge balance}\label{sec:diveqn}
	In this section, we want to deduce in what sense the divergence equations \eqref{eq:diveqn} hold for a solution of \eqref{eq:WholeSystem} in the sense of Definition \ref{def:WeakSolWholeSys}. This is much more difficult than in \cite[Lemma 4.2.]{Guo93} since we consider these divergence equations on whole $\R^3$ instead of $\Omega$. The weak formulation of \eqref{eq:diveqn} is
	\begin{subequations}\label{eq:diveqnweak}
		\begin{align}
		0&=\int_0^{\T}\int_{\R^3}\left(\varepsilon E\cdot\partial_x\varphi+4\pi\rho\varphi\right)\,dxdt,\label{eq:diveqnweakD}\\
		0&=\int_0^{\T}\int_{\R^3}\mu H\cdot\partial_x\varphi\,dxdt\label{eq:diveqnweakB}
		\end{align}
	\end{subequations}
	for all $\varphi\in C_c^\infty\left(\left]0,\T\right[\times\R^3\right)$. Obviously, \eqref{eq:diveqnweak} is equivalent to \eqref{eq:diveqn} be satisfied on $\I\times\R^3$ in the sense of distributions.
	
	For \eqref{eq:diveqn} should propagate in time, we have to demand that \eqref{eq:diveqn} holds initially as a constraint on the initial data, that is to say
	\begin{align*}
	\div\left(\varepsilon\mathring E\right)=4\pi\mathring\rho,\quad\div\left(\mu\mathring H\right)=0
	\end{align*}
	on $\R^3$ in the sense of distributions, or, equivalently,
	\begin{subequations}
		\begin{align}
		0&=\int_{\R^3}\left(\varepsilon\mathring E\cdot\partial_x\xi+4\pi\mathring\rho\xi\right)\,dx,\\
		0&=\int_{\R^3}\mu\mathring H\cdot\partial_x\xi\,dx\label{eq:diveqninitweakB}
		\end{align}
	\end{subequations}
	for all $\xi\in C_c^\infty\left(\R^3\right)$.
	
	Now let $\left(\left(\f,\f_+\right)_\alpha,E,H,j\right)$ be a weak solution of \eqref{eq:WholeSystem} on the time interval $\I$. It is easy to see that \eqref{eq:diveqnweakB} holds: Define
	\begin{align*}
	\vartheta\colon\I\times\R^3\to\R^3,\quad\vartheta\left(t,x\right)=-\int_t^{\T}\partial_x\varphi\left(s,x\right)\,ds.
	\end{align*}
	Clearly, $\vartheta\in\Theta_{\T}$. Hence \eqref{eq:Maxwellweak2} and $\xi=\int_0^{\T}\varphi\left(s,\cdot\right)\,ds$ in \eqref{eq:diveqninitweakB} delivers
	\begin{align*}
	0&=\int_0^{\T}\int_{\R^3}\left(\mu H\cdot\partial_t\vartheta+E\cdot\curl_x\vartheta\right)\,dxdt+\int_{\R^3}\mu\mathring{H}\cdot\vartheta\left(0\right)\,dx\\
	&=\int_0^{\T}\int_{\R^3}\left(\mu H\cdot\partial_x\varphi-E\cdot\int_t^{\T}\curl_x \partial_x\varphi\left(s,x\right)\,ds\right)\,dxdt-\int_{\R^3}\mu\mathring{H}\cdot\partial_x\xi\,dx\\
	&=\int_0^{\T}\int_{\R^3}\mu H\cdot\partial_x\varphi\,dxdt
	\end{align*}
	and we are done.
	
	As for \eqref{eq:diveqnweakD}, we have to exploit local conservation of charge. Consequently, we have to determine what $\rho$ is and have to use the Vlasov equations (their weak form, more precisely). Therefore, we have to make use of \eqref{eq:Vlasovweak} in order to put the internal charge density into play. However, the test functions there have to satisfy $\psi\in\Psi_{\T}$ but a test function of \eqref{eq:diveqnweakD} does not depend on $v$. Consequently, we, on the one hand, have to consider a cut-off in momentum space, and, on the other hand, have to show that \eqref{eq:Vlasovweak} also holds if the support of $\psi$ is not away from $\gamma_{\T}^0$ or $\left\{0\right\}\times\partial\Omega\times\R^3$. For the latter one, the following technical lemma is useful. There and throughout the rest of this section, we assume that $\Omega\subset\R^3$ is a bounded domain such that $\partial\Omega$ is of class $C^1\cap W^{2,\infty}$. Here, $\partial\Omega$ being of class $C^1\cap W^{2,\infty}$ means that it is of class $C^1$ and all local flattenings are locally of class $W^{2,\infty}$.
	\begin{lemma}\label{lma:approxpsi}
		Let $1\leq p<2$ and $\psi\in C^1\left(\I\times\R^3\times\R^3\right)$ with $\supp\psi\subset\left[0,\T\right[\times\R^3\times\R^3$ compact. Then there is a sequence $\left(\psi_k\right)\subset\Psi_{\T}$ such that
		\begin{align}\label{eq:approxpsi}
		\left\|\psi_k-\psi\right\|_{W^{1,p_t2_x1_v}\left(\I\times\Omega\times\R^3\right)}\to 0
		\end{align}
		for $k\to\infty$ and there is $0<r<\infty$ such that $\psi$ and all $\psi_k$ vanish for $t\geq r$. Here,
		\begin{align*}
		\left\|h\right\|_{W^{1,p_t2_x1_v}\left(\I\times\Omega\times\R^3\right)}:=\left(\int_0^{\T}\left(\int_{\Omega}\left(\int_{\R^3}\left(\left|h\right|+\left|\partial_th\right|+\left|\partial_xh\right|+\left|\partial_vh\right|\right)\,dv\right)^2\,dx\right)^{\frac{p}{2}}\,dt\right)^{\frac{1}{p}}.
		\end{align*}
	\end{lemma}
	\begin{proof}
		First, we extend $\psi$ to a $C^1$-function on $\R\times\R^3\times\R^3$ such that $\supp\psi\subset\left]-\T,\T\right[\times\R^3\times\R^3$ is compact (which can be achieved since the hyperplane where $t=0$ is smooth).
		
		By assumption about $\partial\Omega$, for each $x\in\partial\Omega$ there exist open sets $\tilde U_x,\tilde U_x'\subset\R^3$ with $x\in\tilde U_x$ and a $C^1$-diffeomorphism $F^x\colon\tilde U_x\to\tilde U_x'$, that has the property $F^x\in W_\loc^{2,\infty}\left(\tilde U_x;\tilde U_x'\right)$, such that $F^x\left(\tilde U_x\cap\partial\Omega\right)=\tilde U_x'\cap\left(\R^2\times\left\{0\right\}\right)$. For any $x\in\partial\Omega$ we choose an open set $U_x\subset\R^3$ such that $x\in U_x$ and $U_x\subset\subset\tilde U_x$ (here, $A\subset\subset B$ is shorthand for $A$ bounded and $\overline A\subset B$). Then $\partial\Omega\subset\bigcup_{x\in\partial\Omega}U_x$, whence there are a finite number of points, say $x_i\in\partial\Omega$, $i=1,\dots m$, such that $\partial\Omega\subset\bigcup_{i=1}^mU_i$, since $\partial\Omega$ is compact. Here and in the following, we write $U_i:=U_{x_i}$, $\tilde U_i:=\tilde U_{x_i}$, and $F^i:=F^{x_i}$. Since it holds that $\overline\Omega\setminus\bigcup_{i=1}^mU_i\subset\subset\Omega$, there is an open set $U_0\subset\R^3$ satisfying $\overline\Omega\setminus\bigcup_{i=1}^mU_i\subset\subset U_0\subset\subset\Omega$. Therefore we have $\overline\Omega\subset\bigcup_{i=0}^mU_i$. Finally, we choose an open set $M\subset\R^3$ such that $\overline\Omega\subset M\subset\subset\bigcup_{i=0}^mU_i$.
		
		Now let $\zeta_i$, $i=0,\dots,m$, be a partition of unity on $M$ subordinate to $U_i$, $i=0,\dots,m$, i.e., the $\zeta_i$ are of class $C^\infty$, $0\leq\zeta_i\leq 1$, $\supp\zeta_i\subset U_i$, and $\sum_{i=0}^m\zeta_i=1$ on $M$ (and hence on $\overline\Omega$, in particular). Furthermore, let $\eta\in C^\infty\left(\R\right)$ such that $0\leq\eta\leq 1$, $\eta\left(y\right)=0$ for $\left|y\right|\leq\frac{1}{2}$, and $\eta\left(y\right)=1$ for $\left|y\right|\geq 1$.
		
		Next, for $i=1,\dots,m$ define $G^i\colon U_i\times\R^3\to\R^6$, $G^i\left(x,v\right)=\left(F^i\left(x\right),A^i\left(x\right)v\right)$, where the rows $A_j^i\left(x\right)$, $j=1,2,3$, of $A^i\left(x\right)$ are given by
		\begin{align*}
		A_1^i\left(x\right)=\frac{\nabla F_1^i\left(x\right)\times\nabla F_3^i\left(x\right)}{\left|\nabla F_1^i\left(x\right)\times\nabla F_3^i\left(x\right)\right|},\; A_2^i\left(x\right)=\frac{\nabla F_3^i\left(x\right)\times\left(\nabla F_1^i\left(x\right)\times\nabla F_3^i\left(x\right)\right)}{\left|\nabla F_3^i\left(x\right)\times\left(\nabla F_1^i\left(x\right)\times\nabla F_3^i\left(x\right)\right)\right|},\;
		A_3^i\left(x\right)=\frac{\nabla F_3^i\left(x\right)}{\left|\nabla F_3^i\left(x\right)\right|}.
		\end{align*}
		Note that the rows are orthogonal and have length one, and that $A^i$ is of class $C\cap W^{1,\infty}$ on $U_i$ since $F^i$ is of class $C^1\cap W^{2,\infty}$ on $U_i$, $\det DF^i\neq 0$ on $\tilde U_i$, and hence the denominators in $A^i\left(x\right)$ are bounded away from zero on $U_i$ because of $U_i\subset\subset\tilde U_i$. Therefore, $G^i$ is of class $C\cap W^{1,\infty}$ on $U_i\times B_R$ for any $R>0$.
		
		The key idea is that, for any $\left(x,v\right)\in U_i\times\R^3$, $x\in\partial\Omega$ is equivalent to $G_3^i\left(x,v\right)=0$, and, moreover, $\left(x,v\right)\in\tilde\gamma^0$ is equivalent to $G_3^i\left(x,v\right)=G_6^i\left(x,v\right)=0$, since $n\left(x\right)$ and $\nabla F_3^i\left(x\right)$ are parallel (and both non-zero). Thus, since the supports of the approximating functions $\psi_k$ shall be away from $\gamma_{\T}^0$ and $\left\{0\right\}\times\partial\Omega\times\R^3$, it is natural to consider the following $C^\infty$-function in the variables $\left(t,G\right)$, that cuts off a region near the two sets where $G_3=G_6=0$ and where $t=G_3=0$:
		\begin{align*}
		\eta_k\colon\R\times\R^6\to\R,\quad\eta_k\left(t,G\right)=\eta\left(k^2\left(G_3^2+G_6^2\right)\right)\eta\left(k^2\left(t^2+G_3^2\right)\right).
		\end{align*}
		For $k\in\N$ we then define
		\begin{align*}
		\tilde\psi_k\colon\R\times\R^3\times\R^3\to\R,\quad\tilde\psi_k\left(t,x,v\right)=\zeta_0\left(x\right)\psi\left(t,x,v\right)+\sum_{i=1}^m\zeta_i\left(x\right)\psi\left(t,x,v\right)\eta_k^{G^i}\left(t,x,v\right),
		\end{align*}
		where
		\begin{align*}
		\eta_k^{G^i}\colon\R\times U_i\times\R^3\to\R,\quad\eta_k^{G^i}\left(t,x,v\right)=\eta_k\left(t,G^i\left(x,v\right)\right).
		\end{align*}
		We should mention that, because of $\zeta_i\in C_c^\infty\left(U_i\right)$, $i=0,\dots,m$, the $i$-th summand is (by definition) zero if $x\notin U_i$. Note that we can apply the chain rule for $\eta_k^{G^i}$ since $\eta_k$ is smooth and $G^i\in W^{1,1}\left(U_i\times B_R\right)$ for any $R>0$. Therefore, $\tilde\psi_k$ is of class $C\cap W^{1,\infty}$.
		
		First we show that \eqref{eq:approxpsi} holds for $\tilde\psi_k$ (instead of $\psi_k$). By $\sum_{i=0}^m\zeta_i=1$ on $\overline\Omega$ we have
		\begin{align}\label{eq:esttildepsik}
		\left\|\tilde\psi_k-\psi\right\|_{W^{1,p_t2_x1_v}\left(\I\times\Omega\times\R^3\right)}&\leq\sum_{i=1}^m\left\|\zeta_i\psi\left(\eta_k^{G^i}-1\right)\right\|_{W^{1,p_t2_x1_v}\left(\left]0,R\right[\times U_i\times B_R\right)}\nonumber\\
		&\leq C\sum_{i=1}^m\left\|\eta_k^{G^i}-1\right\|_{W^{1,p_t2_x1_v}\left(\left]0,R\right[\times U_i\times B_R\right)},
		\end{align}
		where $C>0$ depends on the (finite) $C_b^1$-norms of $\psi$ (and $\zeta_i$) and where $R>0$ is chosen such that $\psi$ vanishes for $t\geq R$ or $\left|v\right|\geq R$. For fixed $i\in\left\{1,\dots,m\right\}$ and $\left(t,x,v\right)\in\R\times U_i\times\R^3$ there hold the implications
		\begin{align*}
		\eta_k^{G^i}\left(t,x,v\right)\neq 1&\Rightarrow k^2\left(G_3^i\left(x,v\right)^2+G_6^i\left(x,v\right)^2\right)\leq 1\lor k^2\left(t^2+G_3^i\left(x,v\right)^2\right)\leq 1\\
		&\Rightarrow\left|F_3^i\left(x\right)\right|\leq k^{-1}\land\left(\left|G_6^i\left(x,v\right)\right|\leq k^{-1}\lor\left|t\right|\leq k^{-1}\right).
		\end{align*}
		Therefore we have, recalling that $0\leq\eta\leq 1$,
		\begin{align*}
		&\left(\int_0^R\left(\int_{U_i}\left(\int_{B_R}\left|\eta_k^{G^i}-1\right|\,dv\right)^2\,dx\right)^{\frac{p}{2}}\,dt\right)^{\frac{1}{p}}\\
		&\leq\left(\int_0^R\left(\int_{\left\{x\in U_i\mid\left|F_3^i\left(x\right)\right|\leq k^{-1}\right\}}\left(\int_{\left\{v\in B_R\mid\left|G_6^i\left(x,v\right)\right|\leq k^{-1}\right\}}dv\right)^2\,dx\right)^{\frac{p}{2}}\,dt\right)^{\frac{1}{p}}\\
		&\phantom{=\;}+\left(\int_0^{k^{-1}}\left(\int_{\left\{x\in U_i\mid\left|F_3^i\left(x\right)\right|\leq k^{-1}\right\}}\left(\frac{4\pi}{3}R^3\right)^2\,dx\right)^{\frac{p}{2}}\,dt\right)^{\frac{1}{p}}\\
		&=:I_1^k+I_2^k.
		\end{align*}
		In the following we will heavily make use of the facts that $A^i\left(x\right)$ is orthogonal for any $x\in U_i$, $\left|\det DF^i\right|$ is bounded away from zero on $U_i$, and $F^i\left(U_i\right)$ is bounded. Thus
		\begin{align*}
		I_1^k\leq C\left(\int_0^R\left(\int_{\left\{y\in F^i\left(U_i\right)\mid\left|y_3\right|\leq k^{-1}\right\}}\left(\int_{\left\{w\in B_R\mid\left|w_3\right|\leq k^{-1}\right\}}dw\right)^2\,dy\right)^{\frac{p}{2}}\,dt\right)^{\frac{1}{p}}\leq Ck^{-\frac{3}{2}}\to 0
		\end{align*}
		for $k\to\infty$. Here and in the following, $C$ denotes a positive, finite constant that may depend on $p$, $R$, and $F^i$. Similarly,
		\begin{align*}
		I_2^k\leq C\left(\int_0^{k^{-1}}\left(\int_{\left\{y\in F^i\left(U_i\right)\mid\left|y_3\right|\leq k^{-1}\right\}}dy\right)^{\frac{p}{2}}\,dt\right)^{\frac{1}{p}}\leq Ck^{-\frac{1}{2}-\frac{1}{p}}\to 0
		\end{align*}
		for $k\to\infty$. Next we turn to the derivatives and start with the $t$-derivative. By
		\begin{align*}
		\partial_t\eta_k^{G^i}\left(t,x,v\right)=2k^2t\eta\left(k^2\left(G_3^i\left(x,v\right)^2+G_6^i\left(x,v\right)^2\right)\right)\eta'\left(k^2\left(t^2+G_3^i\left(x,v\right)^2\right)\right)
		\end{align*}
		we have
		\begin{align*}
		\left|\partial_t\eta_k^{G^i}\left(t,x,v\right)\right|\leq Ck^2t
		\end{align*}
		and
		\begin{align*}
		\partial_t\eta_k^{G^i}\left(t,x,v\right)\neq 0\Rightarrow k^2\left(t^2+G_3^i\left(x,v\right)^2\right)\leq 1\Rightarrow\left|t\right|\leq k^{-1}\land\left|F_3^i\left(x\right)\right|\leq k^{-1}.
		\end{align*}
		Hence
		\begin{align*}
		&\left(\int_0^R\left(\int_{U_i}\left(\int_{B_R}\left|\partial_t\eta_k^{G^i}\right|\,dv\right)^2\,dx\right)^{\frac{p}{2}}\,dt\right)^{\frac{1}{p}}\\
		&\leq Ck^2\left(\int_0^{k^{-1}}\left(\int_{\left\{x\in U_i\mid\left|F_3^i\left(x\right)\right|\leq k^{-1}\right\}}\left(\int_{B_R}t\,dv\right)^2\,dx\right)^{\frac{p}{2}}\,dt\right)^{\frac{1}{p}}\\
		&\leq Ck^2\left(\int_0^{k^{-1}}\left(\int_{\left\{y\in F^i\left(U_i\right)\mid\left|y_3\right|\leq k^{-1}\right\}}t^2\,dy\right)^{\frac{p}{2}}\,dt\right)^{\frac{1}{p}}\leq Ck^{\frac{3}{2}}\left(\int_0^{k^{-1}}t^p\,dt\right)^{\frac{1}{p}}=Ck^{\frac{1}{2}-\frac{1}{p}},
		\end{align*}
		which converges to $0$ for $k\to\infty$ by $p<2$. This procedure can be performed for the $x$- and $v$-derivatives accordingly, where one needs that $G^i$ is of class $W^{1,\infty}$ on $U_i\times B_R$, resulting in
		\begin{align*}
		\left(\int_0^R\left(\int_{U_i}\left(\int_{B_R}\left|\partial_{x_j}\eta_k^{G^i}\right|\,dv\right)^2\,dx\right)^{\frac{p}{2}}\,dt\right)^{\frac{1}{p}}&\leq Ck^{\frac{1}{2}-\frac{1}{p}}+Ck^{-\frac{1}{2}},\\
		\left(\int_0^R\left(\int_{U_i}\left(\int_{B_R}\left|\partial_{v_j}\eta_k^{G^i}\right|\,dv\right)^2\,dx\right)^{\frac{p}{2}}\,dt\right)^{\frac{1}{p}}&\leq Ck^{-\frac{1}{2}}
		\end{align*}
		for $j=1,2,3$. Altogether we have shown that
		\begin{align*}
		\lim_{k\to\infty}\left\|\eta_k^{G^i}-1\right\|_{W^{1,p_t2_x1_v}\left(\left]0,R\right[\times U_i\times B_R\right)}=0
		\end{align*}
		for any $i=1,\dots,m$ and thus
		\begin{align}\label{eq:tildepsiklimit}
		\lim_{k\to\infty}\left\|\tilde\psi_k-\psi\right\|_{W^{1,p_t2_x1_v}\left(\I\times\Omega\times\R^3\right)}=0
		\end{align}
		by \eqref{eq:esttildepsik}.
		
		The next step is to show that, for each $k\in\N$, the support of $\tilde\psi_k$ is away from $\gamma_{\T}^0$ and $\left\{0\right\}\times\partial\Omega\times\R^3$. As for $\gamma_{\T}^0$, assume the contrary, i.e., $\dist\left(\supp\tilde\psi_k,\gamma_{\T}^0\right)=0$. Then we find sequences $\left(\tilde t_l,\tilde x_l,\tilde v_l\right)_l\subset\gamma_{\T}^0$ and $\left(t_l,x_l,v_l\right)_l\subset\R\times\R^3\times\R^3$ such that $\tilde\psi_k\left(t_l,x_l,v_l\right)\neq 0$ for all $l\in\N$ and
		\begin{align*}
		\lim_{l\to\infty}\left|\left(\tilde t_l,\tilde x_l,\tilde v_l\right)-\left(t_l,x_l,v_l\right)\right|=0.
		\end{align*}
		By compactness of $\supp\tilde\psi_k\subset\supp\psi$, both sequences are bounded, whence we may assume without loss of generality that both sequences converge to the same limit, say $\left(t,x,v\right)\in\R\times\R^3\times\R^3$. Since $\tilde\gamma^0$ is closed and $\tilde t_l\geq 0$ for $l\in\N$, we have $\left(x,v\right)\in\tilde\gamma^0$ and $t\geq 0$. By $\dist\left(x,U_0\right)>0$ and since $\bigcup_{i=1}^mU_i$ is an open cover of $\partial\Omega$, we may also assume that
		\begin{align}\label{eq:xlinUi}
		x_l\in\bigcup_{i\in I}U_i\setminus\bigcup_{i\in\left\{0,\dots,m\right\}\setminus I}U_i,
		\end{align}
		where $I:=\left\{i\in\left\{1,\dots,m\right\}\mid x\in U_i\right\}$ (for $l$ large, at least). Now take $i\in I$. Since $G^i$ is continuous and since $G_3^i\left(x,v\right)=G_6^i\left(x,v\right)=0$ by $\left(x,v\right)\in\tilde\gamma^0$, we have
		\begin{align*}
		\lim_{l\to\infty}G_3^i\left(x_l,v_l\right)=\lim_{l\to\infty}G_6^i\left(x_l,v_l\right)=0
		\end{align*}
		and hence
		\begin{align*}
		G_3^i\left(x_l,v_l\right)^2+G_6^i\left(x_l,v_l\right)^2\leq\frac{1}{2}
		\end{align*}
		for $l$ large. But then $\eta_k^{G^i}\left(t_l,x_l,v_l\right)=0$ and therefore by \eqref{eq:xlinUi}
		\begin{align*}
		0\neq\tilde\psi_k\left(t_l,x_l,v_l\right)=\sum_{i\in I}\zeta_i\left(x_l\right)\psi\left(t_l,x_l,v_l\right)\eta_k^{G^i}\left(t_l,x_l,v_l\right)=0,
		\end{align*}
		which is a contradiction. As for $\left\{0\right\}\times\partial\Omega\times\R^3$, the proof works completely analogously.
		
		There only remains one problem: The approximating functions are only of class $C\cap W^{1,\infty}$ with compact support and not of class $C^\infty$ as desired (which corresponds to the fact that $\partial\Omega$ is only of class $C^1\cap W^{2,\infty}$ and not necessarily smooth). To this end, take a Friedrich's mollifier $\omega\in C_c^\infty\left(\R^7\right)$, $\supp\omega\subset B_1$, $\int_{\R^7}\omega\,dvdxdt=1$, and denote $\omega_\delta:=\delta^{-7}\omega\left(\frac{\cdot}{\delta}\right)$ for $\delta>0$. By $\tilde\psi_k\in H^1\left(\R^7\right)$, we know that $\omega_\delta\ast\tilde\psi_k$ converges to $\tilde\psi_k$ for $\delta\to 0$ in $H^1\left(\R^7\right)$. Moreover, since $\supp\tilde\psi_k\subset\left]-\T,\T\right[\times\R^3\times\R^3$, $\dist\left(\supp\tilde\psi_k,\gamma_{\T}^0\right),\dist\left(\supp\tilde\psi_k,\left\{0\right\}\times\partial\Omega\times\R^3\right)>0$, these properties also hold for $\omega_\delta\ast\tilde\psi_k$ instead of $\tilde\psi_k$ if $\delta$ is small enough. Choose $0<\delta_k\leq 1$ such small and such that 
		\begin{align*}
		\left\|\omega_{\delta_k}\ast\tilde\psi_k-\tilde\psi_k\right\|_{H^1\left(\R^7\right)}\leq\frac{1}{k}.
		\end{align*}
		By $p<2$, this implies
		\begin{align*}
		\left\|\omega_{\delta_k}\ast\tilde\psi_k-\tilde\psi_k\right\|_{W^{1,p_t2_x1_v}\left(\left]0,R\right[\times\Omega\times B_{R+1}\right)}\leq\frac{C}{k}
		\end{align*}
		where $C>0$ depends on $p$, $\Omega$, and $R$. After combining this with \eqref{eq:tildepsiklimit}, noting that $\tilde\psi_k$ and $\psi$ vanish for $t\geq R$ or $\left|v\right|\geq R$ and $\omega_{\delta_k}\ast\tilde\psi_k$ for $t\geq R+1$ (which implies the existence of $r$ as asserted) or $\left|v\right|\geq R+1$, and setting
		\begin{align*}
		\psi_k:=\omega_{\delta_k}\ast\tilde\psi_k\Big|_{\I\times\overline\Omega\times\R^3}\in\Psi_{\T},
		\end{align*}
		we are finally done.
	\end{proof}
	With this lemma, we can extend \eqref{eq:Vlasovweak} to test functions $\psi$ whose supports do not necessarily have to be away from $\gamma_{\T}^0$ and $\left\{0\right\}\times\partial\Omega\times\R^3$ under a condition on the integrability of the solution.
	\begin{lemma}\label{lma:weakVlpsi1}
		For fixed $\alpha\in\left\{1,\dots,N\right\}$ let $\f\in L_\lt^\infty\left(\I\times\Omega\times\R^3\right)$, $\f_+\in L_\lt^\infty\left(\gamma_{\T}^+\right)$, $\left(E,H\right)\in L_\lt^q\left(\I;L^2\left(\R^3;\R^6\right)\right)$ for some $q>2$, $\mathcal K_\alpha\colon L_\lt^\infty\left(\gamma_{\T}^+\right)\to L_\lt^\infty\left(\gamma_{\T}^-\right)$, $\g\in L_\lt^\infty\left(\gamma_{\T}^-\right)$, $\mathring\f\in L^\infty\left(\Omega\times\R^3\right)$ such that Definition \ref{def:WeakSolWholeSys} (ii) is satisfied. Furthermore, let $\psi\in C^\infty\left(\I\times\R^3\times\R^3\right)$ with $\supp\psi\subset\left[0,\T\right[\times\R^3\times\R^3$ compact. Then \eqref{eq:Vlasovweak} still holds for $\psi$.
	\end{lemma}
	\begin{proof}
		Let $1\leq p<2$ satisfy $\frac{1}{p}+\frac{1}{q}=1$. In accordance with Lemma \ref{lma:approxpsi}, let $\left(\psi_k\right)\subset\Psi_{\T}$ approximate $\psi$ with respect to the $W^{1,p_t2_x1_v}$-norm, $0<r<\infty$ such that $\psi$ and all $\psi_k$ vanish for $t\geq r$, and define $R:=\min\left\{r,\T\right\}$. By assumption, \eqref{eq:Vlasovweak} holds for $\psi_k$ for all $k\in\N$. Hence there remains to show that we can pass to the limit $k\to\infty$ in \eqref{eq:Vlasovweak}. First, we have
		\begin{align*}
		&\left|\int_0^{\T}\int_\Omega\int_{\R^3}\left(\partial_t\psi_k-\partial_t\psi\right)\f\,dvdxdt\right|\leq\left\|\psi_k-\psi\right\|_{W^{1,1}\left(\left[0,R\right]\times\Omega\times\R^3\right)}\left\|\f\right\|_{L^\infty\left(\left[0,R\right]\times\Omega\times\R^3\right)}\\
		&\leq C\left(R,\Omega,p,\f\right)\left\|\psi_k-\psi\right\|_{W^{1,p_t2_x1_v}\left(\left[0,R\right]\times\Omega\times\R^3\right)}\to 0
		\end{align*}
		for $k\to\infty$, since $R$ is finite and $\Omega$ is bounded. Similarly,
		\begin{align*}
		\lim_{k\to\infty}\left|\int_0^{\T}\int_\Omega\int_{\R^3}\left(\v\cdot\partial_x\psi_k-\v\cdot\partial_x\psi\right)\f\,dvdxdt\right|=0
		\end{align*}
		by $\left|\v\right|\leq 1$. Next,
		\begin{align*}
		&\left|\int_0^{\T}\int_\Omega\int_{\R^3}\left(E+\v\times H\right)\cdot\left(\partial_v\psi_k-\partial_v\psi\right)\f\,dvdxdt\right|\\
		&\leq\left\|\f\right\|_{L^\infty\left(\left[0,R\right]\times\Omega\times\R^3\right)}\int_0^R\int_\Omega\left(\left|E\right|+\left|H\right|\right)\int_{\R^3}\left|\partial_v\psi_k-\partial_v\psi\right|\,dvdxdt\\
		&\leq C\left(\f\right)\int_0^R\left(\int_\Omega\left(\left|E\right|^2+\left|H\right|^2\right)\,dx\right)^{\frac{1}{2}}\left(\left(\int_{\R^3}\left|\partial_v\psi_k-\partial_v\psi\right|\,dv\right)^2\,dx\right)^{\frac{1}{2}}\,dt\\
		&\leq C\left(\f\right)\left\|\left(E,H\right)\right\|_{L^q\left(\left[0,R\right];L^2\left(\R^3;\R^6\right)\right)}\left(\int_0^R\left(\left(\int_{\R^3}\left|\partial_v\psi_k-\partial_v\psi\right|\,dv\right)^2\,dx\right)^{\frac{p}{2}}\,dt\right)^{\frac{1}{p}}\to 0
		\end{align*}
		for $k\to\infty$. Note that this was the crucial estimate, for which we essentially needed the convergence of $\psi_k$ to $\psi$ in the $W^{1,p_t2_x1_v}$-norm. As for the boundary terms on $\gamma_{\T}^\pm$, we first have
		\begin{align*}
		\int_{\partial\Omega}\left|\psi_k-\psi\right|\,dS_x\leq C\left(\Omega\right)\int_\Omega\left(\left|\psi_k-\psi\right|+\left|\partial_x\psi_k-\partial_x\psi\right|\right)\,dx
		\end{align*}
		for any $t\in\I$, $v\in\R^3$, since $\Omega$ is bounded and $\partial\Omega$ of class $C^1$. Therefore by $\left|n\left(x\right)\cdot\v\right|\leq 1$,
		\begin{align*}
		\left|\int_{\gamma_{\T}^+}\left(\psi_k-\psi\right)\f_+\,d\gamm\right|\leq C\left(\Omega\right)\left\|\psi_k-\psi\right\|_{W^{1,1}\left(\left[0,R\right]\times\Omega\times\R^3\right)}\left\|\f_+\right\|_{L^\infty\left(\gamma_R^+\right)}\to 0
		\end{align*}
		for $k\to\infty$. Similarly,
		\begin{align*}
		&\left|\int_{\gamma_{\T}^-}\left(\psi_k-\psi\right)\left(\mathcal K_\alpha\f_++\g\right)\,d\gamm\right|\\
		&\leq C\left(\Omega\right)\left\|\psi_k-\psi\right\|_{W^{1,1}\left(\left[0,R\right]\times\Omega\times\R^3\right)}\left(\left\|\mathcal K_\alpha\f_+\right\|_{L^\infty\left(\gamma_R^-\right)}+\left\|\g\right\|_{L^\infty\left(\gamma_R^-\right)}\right)\to 0
		\end{align*}
		for $k\to\infty$. Lastly, by
		\begin{align*}
		0=\psi_k\left(R,x,v\right)-\psi\left(R,x,v\right)=\psi_k\left(0,x,v\right)-\psi\left(0,x,v\right)+\int_0^R\left(\partial_t\psi_k\left(t,x,v\right)-\partial_t\psi\left(t,x,v\right)\right)\,dt
		\end{align*}
		for any $x\in\Omega$, $v\in\R^3$, there holds
		\begin{align*}
		\left|\int_\Omega\int_{\R^3}\left(\psi_k\left(0\right)-\psi\left(0\right)\right)\mathring\f\,dvdxdt\right|\leq\left\|\psi_k-\psi\right\|_{W^{1,1}\left(\left[0,R\right]\times\Omega\times\R^3\right)}\left\|\mathring\f\right\|_{L^\infty\left(\Omega\times\R^3\right)}\to 0
		\end{align*}
		for $k\to\infty$, and the proof is complete.
	\end{proof}
	The next step is to show that \eqref{eq:Vlasovweak} still holds if $\psi$ does not depend on $v$. This is done via a cut-off procedure in $v$. Note that in the following lemma it is essential that $\f$ is of class $L^1\cap L_\kin^2$ locally in time.
	\begin{lemma}\label{lma:weakVlpsi2}
		Let $\alpha\in\left\{1,\dots,N\right\}$, $\f\in\left(L_\lt^1\cap L_{\kin,\lt}^2\cap L_\lt^\infty\right)\left(\I\times\Omega\times\R^3\right)$, $\f_+\in L_\lt^\infty\left(\gamma_{\T}^+\right)$, $\left(E,H\right)\in L_\lt^q\left(\I;L^2\left(\R^3;\R^6\right)\right)$ for some $q>2$, $\mathcal K_\alpha\colon L_\lt^\infty\left(\gamma_{\T}^+\right)\to L_\lt^\infty\left(\gamma_{\T}^-\right)$, $\g\in L_\lt^\infty\left(\gamma_{\T}^-\right)$, and $\mathring\f\in\left(L^1\cap L^\infty\right)\left(\Omega\times\R^3\right)$ such that Definition \ref{def:WeakSolWholeSys} (ii) is satisfied. Furthermore, let $\psi\in C^\infty\left(\I\times\R^3\right)$ with $\supp\psi\subset\left[0,\T\right[\times\R^3$ compact.
		\begin{enumerate}[label=(\roman*)]
			\item If $\supp\psi\subset\left[0,\T\right[\times\left(\R^3\setminus\partial\Omega\right)$, we have
			\begin{align}\label{eq:Vlasovweakpsinovdepi}
			0&=\int_0^{\T}\int_\Omega\left(\partial_t\psi\int_{\R^3}\f\,dv+\partial_x\psi\cdot\int_{\R^3}\v\f\,dv\right)dxdt+\int_\Omega\psi\left(0\right)\int_{\R^3}\mathring\f\,dvdx.
			\end{align}
			\item\label{lma:weakVlpsi2ii} If, additionally to the given assumptions, $\f_+\in L_\lt^1\left(\gamma_{\T}^+,d\gamm\right)$, $\g\in L_\lt^1\left(\gamma_{\T}^-,d\gamm\right)$, and $\mathcal K_\alpha\colon\left(L_\lt^1\cap L_\lt^\infty\right)\left(\gamma_{\T}^+,d\gamm\right)\to\left(L_\lt^1\cap L_\lt^\infty\right)\left(\gamma_{\T}^-,d\gamm\right)$, but $\psi$ may not vanish on $\partial\Omega$, then \eqref{eq:Vlasovweak} is still satisfied for $\psi$, i.e.,
			\begin{align}\label{eq:Vlasovweakpsinovdepii}
			0&=-\int_0^{\T}\int_\Omega\left(\partial_t\psi\int_{\R^3}\f\,dv+\partial_x\psi\cdot\int_{\R^3}\v\f\,dv\right)dxdt+\int_{\gamma_{\T}^+}\f_+\psi\,d\gamm\nonumber\\
			&\phantom{=\;}-\int_{\gamma_{\T}^-}\left(\K\f_++\g\right)\psi\,d\gamm-\int_\Omega\psi\left(0\right)\int_{\R^3}\mathring\f\,dvdx.
			\end{align}
		\end{enumerate}
	\end{lemma}
	\begin{proof}
		The proof works similarly to the proof of \cite[Lemma 4.2.]{Guo93}. First, consider a test function $\psi$ that may have support on $\partial\Omega$. Take $\eta\in C_c^\infty\left(\R^3\right)$, $0\leq\eta\leq 1$, $\eta=1$ on $B_1$, $\supp\eta\subset B_2$, and let $\eta_m\left(v\right):=\eta\left(\frac{v}{m}\right)$ for $m\in\N$, $v\in\R^3$. Then $\psi_m:=\psi\eta_m\in C^\infty\left(\I\times\R^3\times\R^3\right)$ with $\supp\psi\subset\left[0,\T\right[\times\R^3\times\R^3$ compact. Therefore, \eqref{eq:Vlasovweak} holds for $\psi_m$ by Lemma \ref{lma:weakVlpsi1}. Now we can show that we may pass to the limit $m\to\infty$ in all terms of \eqref{eq:Vlasovweak} but the terms including integrals over $\gamma_{\T}^\pm$. Let $R>0$ such that $\psi$ vanishes for $t\geq R$. First,
		\begin{align*}
		&\left|\int_0^{\T}\int_\Omega\int_{\R^3}\f\partial_t\psi_m\,dvdxdt-\int_0^{\T}\int_\Omega\partial_t\psi\int_{\R^3}\f\,dvdxdt\right|\\
		&\leq\left\|\partial_t\psi\right\|_{L^\infty\left(\I\times\R^3\right)}\int_0^R\int_\Omega\int_{\R^3}\left|\eta_m-1\right|\left|\f\right|\,dvdxdt\overset{m\to\infty}\to 0
		\end{align*}
		by dominated convergence since $\eta_m\to 1$ pointwise for $m\to\infty$ and $\left|\eta_m-1\right|\left|\f\right|\leq\left|\f\right|\in L^1\left(\left[0,R\right]\times\Omega\times\R^3\right)$. Similarly by $\left|\v\right|\leq 1$,
		\begin{align*}
		\lim_{m\to\infty}\int_0^{\T}\int_\Omega\int_{\R^3}\partial_x\psi_m\cdot\v\f\,dvdxdt=\int_0^{\T}\int_\Omega\partial_x\psi\cdot\int_{\R^3}\v\f\,dvdxdt.
		\end{align*}
		By
		\begin{align*}
		\partial_v\psi_m\left(t,x,v\right)=\frac{1}{m}\psi\left(t,x\right)\eta'\left(\frac{v}{m}\right)
		\end{align*}
		and
		\begin{align*}
		\partial_v\psi_m\left(t,x,v\right)\neq 0\Rightarrow m\leq\left|v\right|\leq 2m
		\end{align*}
		for $\left(t,x,v\right)\in\I\times\Omega\times\R^3$, we get the following estimate, which is again the crucial one:
		\begin{align*}
		&\left|\int_0^{\T}\int_\Omega\int_{\R^3}\left(E+\v\times H\right)\f\cdot\partial_v\psi_m\,dvdxdt\right|\\
		&\leq\left\|\psi\right\|_{L^\infty\left(\I\times\Omega\right)}\left\|\eta'\right\|_{L^\infty\left(B_2\right)}\int_0^R\int_\Omega\left(\left|E\right|+\left|H\right|\right)\int_{\left\{v\in\R^3\mid m\leq\left|v\right|\leq 2m\right\}}\frac{1}{m}\left|\f\right|\,dvdxdt\\
		&\leq C\left(\psi,\eta\right)\left\|\left(E,H\right)\right\|_{L^2\left(\left[0,R\right]\times\Omega;\R^6\right)}\left(\int_0^R\int_\Omega\left(\int_{\left\{v\in\R^3\mid m\leq\left|v\right|\leq 2m\right\}}\frac{1}{m}\left|\f\right|\,dv\right)^2\,dxdt\right)^{\frac{1}{2}}\\
		&\leq C\left(\psi,\eta,E,H\right)\left(\int_0^R\int_\Omega\int_{\left\{v\in\R^3\mid m\leq\left|v\right|\leq 2m\right\}}\frac{\frac{4\pi}{3}\left(8m^3-m^3\right)}{m^2}\left|\f\right|^2\,dv\,dxdt\right)^{\frac{1}{2}}\\
		&\leq C\left(\psi,\eta,E,H\right)\left(\int_0^R\int_\Omega\int_{\left\{v\in\R^3\mid m\leq\left|v\right|\leq 2m\right\}}\vo\left|\f\right|^2\,dv\,dxdt\right)^{\frac{1}{2}}\to 0
		\end{align*}
		for $m\to\infty$, since the last integral converges to zero by $\f\in L_\kin^2\left(\left[0,R\right]\times\Omega\times\R^3\right)$. As for the term including the initial data, we see that
		\begin{align*}
		\left|\int_\Omega\int_{\R^3}\psi_m\left(0\right)\mathring\f\,dvdx-\int_\Omega\psi\left(0\right)\int_{\R^3}\mathring\f\,dvdx\right|\leq\left\|\psi\left(0\right)\right\|_{L^\infty\left(\Omega\right)}\int_\Omega\int_{\R^3}\left|\eta_m-1\right|\left|\mathring\f\right|\,dvdx\to 0
		\end{align*}
		for $m\to\infty$ as well by dominated convergence and $\mathring\f\in L^1\left(\Omega\times\R^3\right)$.
		
		Now if $\supp\psi\subset\left[0,\T\right[\times\left(\R^3\setminus\partial\Omega\right)$, then $\psi_m$ vanishes on $\partial\Omega$, too, and there vanish the integrals over $\gamma_{\T}^\pm$ for $\psi_m$ appearing in \eqref{eq:Vlasovweak}. Hence, \eqref{eq:Vlasovweakpsinovdepi} is satisfied.
		
		If the additional assumptions of \ref{lma:weakVlpsi2ii} hold, but $\psi$ may not vanish on $\partial\Omega$, we consider the integrals over $\gamma_{\T}^\pm$:
		\begin{align*}
		\left|\int_{\gamma_{\T}^+}\f_+\psi_m\,d\gamm-\int_{\gamma_{\T}^+}\f_+\psi\,d\gamm\right|\leq\left\|\psi\right\|_{L^\infty\left(\I\times\partial\Omega\right)}\int_{\gamma_R^+}\left|\eta_m-1\right|\left|\f_+\right|\,d\gamm\overset{m\to\infty}\to 0
		\end{align*}
		and similarly
		\begin{align*}
		&\left|\int_{\gamma_{\T}^-}\left(\mathcal K_\alpha\f_++\g\right)\psi_m\,d\gamm-\int_{\gamma_{\T}^-}\left(\mathcal K_\alpha\f_++\g\right)\psi\,d\gamm\right|\\
		&\leq\left\|\psi\right\|_{L^\infty\left(\I\times\partial\Omega\right)}\int_{\gamma_R^-}\left|\eta_m-1\right|\left(\left|\mathcal K_\alpha\f_+\right|+\left|\g\right|\right)\,d\gamm\overset{m\to\infty}\to 0
		\end{align*}
		by dominated convergence and $\f_+\in L^1\left(\gamma_R^+,d\gamm\right)$, $\mathcal K_\alpha\f_+,\g\in L^1\left(\gamma_R^-,d\gamm\right)$. Therefore we obtain \eqref{eq:Vlasovweakpsinovdepii}.
	\end{proof}
	In the following, we denote
	\begin{align*}
	\rho^\inte:=\suma\e\int_{\R^3}\f\,dv,\quad j^\inte:=\suma\e\int_{\R^3}\v\f\,dv
	\end{align*}
	and extend these functions by zero for $x\notin\Omega$.
	
	Equations \eqref{eq:Vlasovweakpsinovdepi} and \eqref{eq:Vlasovweakpsinovdepii} reflect the principle of local conservation of the internal charge and imply a global charge balance after an integration:
	\begin{corollary}\label{cor:conservcharge}
		Let the assumptions of Lemma \ref{lma:weakVlpsi2} hold for all $\alpha\in\left\{1,\dots,N\right\}$.
		\begin{enumerate}[label=(\roman*)]
			\item\label{cor:conservchargei} We have 
			\begin{align*}
			\partial_t\rho^\inte+\div_xj^\inte=0
			\end{align*}
			on $\left]0,\T\right[\times\Omega$ in the sense of distributions.
		\end{enumerate}
		If moreover the additional assumptions of Lemma \ref{lma:weakVlpsi2} (ii) are satisfied for all $\alpha\in\left\{1,\dots,N\right\}$, then:
		\begin{enumerate}[label=(\roman*),resume]
			\item\label{cor:conservchargeii} There holds
			\begin{align}\label{eq:locconservintcharge}
			\partial_t\rho^\inte+T_{\partial\Omega}+\div_xj^\inte=0
			\end{align}
			on $\left]0,\T\right[\times\R^3$ in the sense of distributions. Here, the distribution $T_{\partial\Omega}$ describes the boundary processes via
			\begin{align*}
			T_{\partial\Omega}\psi=\suma\e\left(\int_{\gamma_{\T}^+}\f_+\psi\,d\gamm-\int_{\gamma_{\T}^-}\left(\K\f_++\g\right)\psi\,d\gamm\right).
			\end{align*}
			\item\label{cor:conservchargeiii} For almost all $t\in\I$ we have
			\begin{align*}
			\int_\Omega\rho^\inte\left(t,x\right)\,dx=\int_\Omega\mathring\rho^\inte\,dx-\suma\e\left(\int_{\gamma_t^+}\f_+\,d\gamm-\int_{\gamma_t^-}\left(\K\f_++\g\right)\,d\gamm\right).
			\end{align*}
		\end{enumerate}
	\end{corollary}
	\begin{proof}
		As for \ref{cor:conservchargei} and \ref{cor:conservchargeii}, simply multiply \eqref{eq:Vlasovweakpsinovdepi} and \eqref{eq:Vlasovweakpsinovdepii} with $\e$ and sum over $\alpha$. As for \ref{cor:conservchargeiii}, take $\varphi\in C_c^\infty\left(\left]0,\T\right[\right)$. Choose $\eta\in C_c^\infty\left(\R^3\right)$ with $\eta=1$ on $\overline\Omega$. We define
		\begin{align*}
		\psi\colon\I\times\R^3\to\R,\quad\psi\left(t,x\right)=-\eta\left(x\right)\int_t^{\T}\varphi\,ds.
		\end{align*}
		Then $\psi\in C^\infty\left(\I\times\R^3\right)$ with $\supp\psi\subset\left[0,\T\right[\times\R^3$ compact. Therefore, Lemma \ref{lma:weakVlpsi2} (ii) yields, after summing over $\alpha$,
		\begin{align*}
		0&=\suma\e\left(-\int_0^{\T}\int_\Omega\left(\partial_t\psi\int_{\R^3}\f\,dv+\partial_x\psi\cdot\int_{\R^3}\v\f\,dv\right)dxdt+\int_{\gamma_{\T}^+}\f_+\psi\,d\gamm\right.\nonumber\\
		&\phantom{=\suma\e\;}\left.-\int_{\gamma_{\T}^-}\left(\K\f_++\g\right)\psi\,d\gamm-\int_\Omega\psi\left(0\right)\int_{\R^3}\mathring\f\,dvdx\right)\\
		&=-\int_0^{\T}\varphi\int_\Omega\rho^\inte\,dxdt+\int_0^{\T}\varphi\int_\Omega\mathring\rho^\inte\,dxds\\
		&\phantom{=\;}+\suma\e\left(-\int_0^{\T}\int_{\partial\Omega}\int_{\left\{v\in\R^3\mid n\left(x\right)\cdot v>0\right\}}\f_+\left(t,x,v\right)\int_t^{\T}\varphi\left(s\right)\,ds\,n\left(x\right)\cdot\v\,dvdS_xdt\right.\\
		&\phantom{=\;}\left.-\int_0^{\T}\int_{\partial\Omega}\int_{\left\{v\in\R^3\mid n\left(x\right)\cdot v<0\right\}}\left(\K\f_++\g\right)\left(t,x,v\right)\int_t^{\T}\varphi\left(s\right)\,ds\,n\left(x\right)\cdot\v\,dvdS_xdt\right)\\
		&=-\int_0^{\T}\varphi\left(\int_\Omega\rho^\inte\,dx-\int_\Omega\mathring\rho^\inte\,dx\right)\,dt\\
		&\phantom{=\;}+\suma\e\left(-\int_0^{\T}\varphi\left(s\right)\int_0^s\int_{\partial\Omega}\int_{\left\{v\in\R^3\mid n\left(x\right)\cdot v>0\right\}}\f_+\left(t,x,v\right)n\left(x\right)\cdot\v\,dvdS_xdtds\right.\\
		&\phantom{=\;}\left.-\int_0^{\T}\varphi\left(s\right)\int_0^s\int_{\partial\Omega}\int_{\left\{v\in\R^3\mid n\left(x\right)\cdot v<0\right\}}\left(\K\f_++\g\right)\left(t,x,v\right)n\left(x\right)\cdot\v\,dvdS_xdtds\right),
		\end{align*}
		from which the assertion follows immediately.
	\end{proof}
	We can finally show the remaining parts of Theorem \ref{thm:reddivE} with the help of Lemma \ref{lma:weakVlpsi2}; the redundancy of $\div_x\left(\mu H\right)=0$ has already been proved. To this end, assume Condition \ref{cond:extchargedens}.
	\begin{proof}[of Theorem \ref{thm:reddivE}]
		First take $\varphi\in C_c^\infty\left(\left]0,\T\right[\times\R^3\right)$ arbitrary. Define
		\begin{align*}
		\psi&\colon\I\times\R^3\to\R,&\psi\left(t,x\right)&=-\int_t^{\T}\varphi\left(s,x\right)\,ds,\\
		\vartheta&\colon\I\times\R^3\to\R^3,&\vartheta\left(t,x\right)&=-\int_t^{\T}\partial_x\varphi\left(s,x\right)\,ds,\\
		\xi&\colon\R^3\to\R,&\xi\left(x\right)&=\int_0^{\T}\varphi\left(s,x\right)\,ds.
		\end{align*}
		Clearly, $\psi\in C^\infty\left(\I\times\R^3\right)$ with $\supp\psi\subset\left[0,\T\right[\times\R^3$ compact, $\vartheta\in\Theta_{\T}$, and $\xi\in C_c^\infty\left(\R^3\right)$. By $\vartheta\in\Theta_{\T}$, there holds \eqref{eq:Maxwellweak1}, i.e.,
		\begin{align}\label{eq:redmax}
		0&=\int_0^{\T}\int_{\R^3}\left(\varepsilon E\cdot\partial_t\vartheta-H\cdot\curl_x\vartheta-4\pi\left(j^\inte+u\right)\cdot\vartheta\right)\,dxdt+\int_{\R^3}\varepsilon\mathring{E}\cdot\vartheta\left(0\right)\,dx\nonumber\\
		&=\int_0^{\T}\int_{\R^3}\left(\varepsilon E\cdot\partial_x\varphi+H\cdot\int_t^{\T}\curl_x\partial_x\varphi\,ds-4\pi\left(j^\inte+u\right)\cdot\vartheta\right)\,dxdt-\int_{\R^3}\varepsilon\mathring{E}\cdot\partial_x\xi\,dx\nonumber\\
		&=\int_0^{\T}\int_{\R^3}\left(\varepsilon E\cdot\partial_x\varphi-4\pi\left(j^\inte+u\right)\cdot\vartheta\right)\,dxdt-\int_{\R^3}\varepsilon\mathring{E}\cdot\partial_x\xi\,dx.
		\end{align}
		By Condition \ref{cond:extchargedens} we have
		\begin{align}\label{eq:redconservextcharge}
		0&=\int_0^{\T}\int_{\R^3}\left(\rho^u\partial_t\psi+u\cdot\partial_x\psi\right)\,dxdt+\int_{\R^3}\mathring\rho^u\psi\left(0\right)\,dx\nonumber\\
		&=\int_0^{\T}\int_{\R^3}\left(\rho^u\varphi+u\cdot\vartheta\right)\,dxdt-\int_{\R^3}\mathring\rho^u\xi\,dx.
		\end{align}
		
		To prove \ref{thm:reddivEi}, assume that $\varphi\in C_c^\infty\left(\left]0,\T\right[\times\left(\R^3\setminus\partial\Omega\right)\right)$. Then we have $\psi\in C^\infty\left(\I\times\R^3\right)$ with $\supp\psi\subset\left[0,\T\right[\times\left(\R^3\setminus\partial\Omega\right)$ compact and Lemma \ref{lma:weakVlpsi2} (i) gives us, after multiplying with $\e$ and summing over $\alpha$,
		\begin{align}\label{eq:redvlasovi}
		0&=\int_0^{\T}\int_\Omega\left(\rho^\inte\partial_t\psi+j^\inte\cdot\partial_x\psi\right)\,dxdt+\int_\Omega\mathring\rho^\inte\psi\left(0\right)\,dx\nonumber\\
		&=\int_0^{\T}\int_\Omega\left(\rho^\inte\varphi+j^\inte\cdot\vartheta\right)\,dxdt-\int_\Omega\mathring\rho^\inte\xi\,dx.
		\end{align}
		Multiplying \eqref{eq:redconservextcharge} and \eqref{eq:redvlasovi} with $4\pi$ and adding them to \eqref{eq:redmax} yields
		\begin{align*}
		\int_0^{\T}\int_{\R^3}\left(\varepsilon E\cdot\partial_x\varphi+4\pi\left(\rho^\inte+\rho^u\right)\varphi\right)\,dx=\int_{\R^3}\left(\varepsilon\mathring E\cdot\partial_x\xi+4\pi\left(\mathring\rho^\inte+\mathring\rho^u\right)\xi\right)\,dx=0
		\end{align*}
		by $\div_x\left(\varepsilon\mathring E\right)=4\pi\left(\mathring\rho^\inte+\mathring\rho^u\right)$ on $\R^3$ in the sense of distributions. Hence, $\div_x\left(\varepsilon E\right)=4\pi\left(\rho^\inte+\rho^u\right)$ on $\left]0,\T\right[\times\left(\R^3\setminus\partial\Omega\right)$ in the sense of distributions.
		
		To prove \ref{thm:reddivEii}, let the additional assumptions stated there hold. The test function $\varphi\in C_c^\infty\left(\left]0,\T\right[\times\R^3\right)$ may now not vanish on $\partial\Omega$. Then we have $\psi\in C^\infty\left(\I\times\R^3\right)$ with $\supp\psi\subset\left[0,\T\right[\times\R^3$ compact and Lemma \ref{lma:weakVlpsi2} (ii) gives us, after multiplying with $\e$ and summing over $\alpha$,
		\begin{align}\label{eq:redvlasovii}
		0&=\int_0^{\T}\int_\Omega\left(\rho^\inte\partial_t\psi+j^\inte\cdot\partial_x\psi\right)\,dxdt-T_{\partial\Omega}\psi+\int_\Omega\mathring\rho^\inte\psi\left(0\right)\,dx\nonumber\\
		&=\int_0^{\T}\left(\rho^\inte\varphi+j^\inte\cdot\vartheta\right)\,dxdt-T_{\partial\Omega}\psi-\int_\Omega\mathring\rho^\inte\xi\,dx.
		\end{align}
		We rewrite $T_{\partial\Omega}\psi$:
		\begin{align*}
		T_{\partial\Omega}\psi&=\suma\e\left(\int_{\gamma_{\T}^+}\f_+\psi\,d\gamm-\int_{\gamma_{\T}^-}\left(\K\f_++\g\right)\psi\,d\gamm\right)\\
		&=\suma\e\left(-\int_0^{\T}\int_{\partial\Omega}\int_{\left\{v\in\R^3\mid n\left(x\right)\cdot v>0\right\}}\f_+\left(t,x,v\right)\int_t^{\T}\varphi\left(s,x\right)\,ds\,n\left(x\right)\cdot\v\,dvdS_xdt\right.\\
		&\phantom{=\;}\left.-\int_0^{\T}\int_{\partial\Omega}\int_{\left\{v\in\R^3\mid n\left(x\right)\cdot v<0\right\}}\left(\K\f_++\g\right)\left(t,x,v\right)\int_t^{\T}\varphi\left(s,x\right)\,ds\,n\left(x\right)\cdot\v\,dvdS_xdt\right)\\
		&=\suma\e\left(-\int_0^{\T}\int_{\partial\Omega}\varphi\left(s,x\right)\int_0^s\int_{\left\{v\in\R^3\mid n\left(x\right)\cdot v>0\right\}}\f_+\left(t,x,v\right)n\left(x\right)\cdot\v\,dvdtdS_xds\right.\\
		&\phantom{=\;}\left.-\int_0^{\T}\int_{\partial\Omega}\varphi\left(s,x\right)\int_0^s\int_{\left\{v\in\R^3\mid n\left(x\right)\cdot v<0\right\}}\left(\K\f_++\g\right)\left(t,x,v\right)n\left(x\right)\cdot\v\,dvdtdS_xds\right)\\
		&=-S_{\partial\Omega}\varphi.
		\end{align*}
		Similarly as before, multiplying \eqref{eq:redconservextcharge} and \eqref{eq:redvlasovii} with $4\pi$ and adding them to \eqref{eq:redmax} yields
		\begin{align*}
		\int_0^{\T}\int_{\R^3}\left(\varepsilon E\cdot\partial_x\varphi+4\pi\left(\rho^\inte+\rho^u\right)\varphi\right)\,dx+4\pi S_{\partial\Omega}\varphi=\int_{\R^3}\left(\varepsilon\mathring E\cdot\partial_x\xi+4\pi\left(\mathring\rho^\inte+\mathring\rho^u\right)\xi\right)\,dx=0.
		\end{align*}
		Hence, $\div_x\left(\varepsilon E\right)=4\pi\left(\rho^\inte+\rho^u+S_{\partial\Omega}\right)$ on $\left]0,\T\right[\times\R^3$ in the sense of distributions.
	\end{proof}
	\begin{remark}
		We discuss some assumptions and give some comments regarding Theorem \ref{thm:reddivE} and Corollary \ref{cor:conservcharge}:
		\begin{itemize}
			\item Clearly, we see by interpolation that $\f\in\left(L_{\kin,\lt}^1\cap L_\lt^\infty\right)\left(\I\times\Omega\times\R^3\right)$ implies $\f\in\left(L_\lt^1\cap L_{\kin,\lt}^2\cap L_\lt^\infty\right)\left(\I\times\Omega\times\R^3\right)$ and $\left(E,H\right)\in L_\lt^\infty\left(\I;L^2\left(\R^3;\R^6\right)\right)$ ensures $\left(E,H\right)\in L_\lt^q\left(\I;L^2\left(\R^3;\R^6\right)\right)$. Hence, the $\f$ and $E$, $H$ of Theorem \ref{thm:Existence} satisfy these assumptions, and Theorem \ref{thm:reddivE} (i), (ii) can be applied. However, the boundary values $\f_+$ constructed there only satisfy $\f_+\in L_\lt^1\left(\gamma_{\T}^+,d\gamm\right)$ for $\alpha=1,\dots,N'$, i.e., the particles are subject to partially absorbing boundary conditions, and not necessarily for $\alpha=N'+1,\dots,N$, i.e., the particles are subject to purely reflecting boundary conditions. Therefore, whether the statement of Theorem \ref{thm:reddivE} (iii) is true for the solution of Theorem \ref{thm:Existence}, remains as an open problem, unless $N'=N$, i.e., all particles are subject to partially absorbing boundary conditions.
			\item Conversely, the assumption $\f_+\in L_\lt^1\left(\gamma_{\T}^+,d\gamm\right)$ is necessary for Theorem \ref{thm:reddivE} (iii) (and for Lemma \ref{lma:weakVlpsi2} (ii)). Otherwise, the integral $\int_{\gamma_{\T}^+}\f_+\psi\,d\gamm$ will not exist in general, since $\psi$ need not vanish on $\partial\Omega$ and does not depend on $v$.
			\item The distribution $S_{\partial\Omega}$ can be interpreted as follows: The terms
			\begin{align*}
			j_{\partial\Omega}^{\mathrm{out}}\left(t,x\right)&:=\suma\e\int_{\left\{v\in\R^3\mid n\left(x\right)\cdot v>0\right\}}\v\f_+\left(t,x,v\right)\,dv,\\
			j_{\partial\Omega}^{\mathrm{in}}\left(t,x\right)&:=\suma\e\int_{\left\{v\in\R^3\mid n\left(x\right)\cdot v<0\right\}}\v\left(\left(\mathcal K_\alpha\f_+\right)\left(t,x,v\right)+\g\left(t,x,v\right)\right)\,dv,
			\end{align*}
			where $\left(t,x\right)\in\I\times\partial\Omega$, can be viewed as the outgoing and incoming boundary current density. Hence $S_{\partial\Omega}$ can be rewritten as
			\begin{align*}
			S_{\partial\Omega}\varphi&=\int_0^{\T}\int_{\partial\Omega}\varphi\left(t,x\right)\int_0^tn\left(x\right)\cdot\left(j_{\partial\Omega}^{\mathrm{out}}\left(s,x\right)+j_{\partial\Omega}^{\mathrm{in}}\left(s,x\right)\right)\,dsdS_xdt.
			\end{align*}
			Thus, $S_{\partial\Omega}$ measures how many particles have left and entered $\Omega$ \textit{up to} time $t$. On the other hand, the distribution $T_{\partial\Omega}$ measures how many particles leave and enter $\Omega$ \textit{at} time $t$ via
			\begin{align*}
			T_{\partial\Omega}\psi=\int_0^{\T}\int_{\partial\Omega}\psi\left(t,x\right)n\left(x\right)\cdot\left(j_{\partial\Omega}^{\mathrm{out}}\left(t,x\right)+j_{\partial\Omega}^{\mathrm{in}}\left(t,x\right)\right)\,dS_xdt.
			\end{align*}
			We easily see that $\partial_t S_{\partial\Omega}=T_{\partial\Omega}$ on $\left]0,\T\right[\times\R^3$ in the sense of distributions, which corresponds to the fact that $T_{\partial\Omega}$ appears as \enquote{a part of $\partial_t\rho$} in \eqref{eq:locconservintcharge} and $S_{\partial\Omega}$ appears as \enquote{a part of $\rho$} in \eqref{eq:divErho}.
			\item The global charge balance, see Corollary \ref{cor:conservcharge} (iii), can similarly been written as follows:
			\begin{align*}
			\int_\Omega\rho^\inte\left(t,x\right)\,dx=\int_\Omega\mathring\rho^\inte\,dx-\int_0^t\int_{\partial\Omega} n\cdot\left(j_{\partial\Omega}^{\mathrm{out}}+j_{\partial\Omega}^{\mathrm{in}}\right)\,dS_xds
			\end{align*}
			for almost all $t\in\I$.
			\item As mentioned in the introduction, in a more realistic model $\varepsilon$ and $\mu$ should depend on $\f$, $E$, and $H$ (maybe even nonlocally) and hence implicitly on time. In this situation, the weak formulation is the same as before, which is stated in Definition \ref{def:WeakSolWholeSys}. If we assume $\varepsilon,\mu\in L_\loc^\infty\left(\I\times\R^3;\R^{3\times 3}\right)$ (and suitably introduce initial values for $\varepsilon,\mu$), viewed as explicit functions of $t$ and $x$, the proofs of Theorem \ref{thm:reddivE} and the lemmas before are still valid, and Theorem \ref{thm:reddivE} remains true.
			\item Lastly, we emphasize that the results of this section hold, under the respective assumptions, for all weak solutions of \eqref{eq:WholeSystem} in the sense of Definition \ref{def:WeakSolWholeSys} and not only for the solutions of Theorem \ref{thm:Existence}.
		\end{itemize}
	\end{remark}
	
	\nocite{*}
	\bibliography{solrvmwec}
	\bibliographystyle{plain}
\end{document}